\documentclass[11pt,reqno,UTF8,twoside]{amsart}
\usepackage{mathrsfs}
\usepackage{amsfonts,amssymb,amsmath,amsthm}
\usepackage{cite}
\usepackage[colorlinks=true,citecolor=red]{hyperref}
\usepackage{titletoc}
\usepackage{geometry}
\usepackage{bm}
\usepackage{indentfirst}
\usepackage{graphicx}
\usepackage{float} 
\usepackage{booktabs}
\usepackage{longtable}
\usepackage{subfigure}
\usepackage{stmaryrd}
\usepackage{enumerate}
\usepackage{tikz}
\usepackage{ulem}
\usepackage{color,soul}
\usepackage{setspace}
\usepackage{stmaryrd}
\usepackage[pagewise]{lineno} 
\allowdisplaybreaks[2]
\usepackage{appendix}
\usepackage{cases}
\usepackage{todonotes}
\usepackage{enumitem}
\hypersetup{linkcolor=blue}

\makeatletter
\@namedef{subjclassname@2020}{\textup{2020} Mathematics Subject 
	Classification}
\makeatother
\usepackage{amssymb}
\usepackage[dvips]{epsfig}
\usepackage{graphicx}
\usepackage{color}
\usepackage{ulem}
\usepackage{tikz}
\usepackage{amsmath}
\usepackage{amssymb}
\usepackage{amsfonts}
\usepackage{amsthm}
\usepackage{tikz}
\usepackage{bm}
\usetikzlibrary{calc}
\usepackage{changes}

\newcommand{\R}{\mathbb R}

\newcommand{\Z}{\mathbb Z}

\usepackage{float}

\newcommand{\rank}{\operatorname{rank}}
\newcommand{\Ran}{\operatorname{range}}

\newcommand{\T}{\mathbb{T}}

\newtheorem{thm}{Theorem}[section]

\newtheorem{lem}[thm]{Lemma}
\newtheorem{prop}[thm]{Proposition}
\newtheorem{cor}[thm]{\bf Corollary}

\newtheorem{rem}[thm]{\bf Remark}

\theoremstyle{definition}
\newtheorem{defn}[thm]{Definition}

\newtheorem{claim}[thm]{\bf Claim}

\theoremstyle{statement}

\numberwithin{equation}{section}

\usepackage{caption}

\usepackage{bm}
\begin{document}
	\title[Monotone quasiperiodic long-range operators]{Arithmetic localization for quasiperiodic long-range operators with monotone potentials}
		\author[Hongyi Cao]{Hongyi Cao}
		\address[H. Cao] {School of Mathematical Sciences,
				Peking University,
				Beijing 100871,
				China}
		\email{chyyy@math.pku.edu.cn}

	\begin{abstract}
We study one-dimensional quasiperiodic  operators with  $\alpha$-H\"older monotone potentials  and  long-range hopping satisfying $|T(x,y)|\leq \varepsilon e^{-\rho |x-y|}$.  Under an  explicit smallness condition on $\varepsilon$,  we  establish  uniform Anderson  localization for such  operators in the regime $\rho>\alpha\beta $, where $\beta$ is  the exponential strength of frequency resonances.  Our proof extends the arithmetic  localization strategy for monotone quasiperiodic Schr\"odinger operators by Jitomirskaya and Kachkovskiy  to the long-range setting, without relying on the Lyapunov exponent.
	\end{abstract}
\date{\today}

	\maketitle

	\section{Introduction}
We study the following  family of one-dimensional  quasiperiodic long-range  operators
	\begin{equation}\label{md} 
		H(\theta)= \operatorname{diag}(v(\theta+x\omega))_{x\in \Z}+\varepsilon T_\phi  \end{equation} 
	acting on $\ell^2(\Z)$ by
	$$ (H(\theta )u)(x)=
	v(\theta+x\omega)u(x)+	\varepsilon\sum_{y\in\Z}\phi_{x-y}u(y), \quad u\in \ell^2(\Z). 
	$$
For the diagonal term, the potential $v$ is a $1$-periodic, real-valued function, $\theta\in \R$ is the phase and $\omega\in\R\setminus\mathbb Q$ is the frequency.  By periodicity,  we may also view $v$, $\theta$ and $\omega$ on 
$\T:=\R/\Z$ when convenient. 
 For the off-diagonal term,  $\varepsilon\geq0$ is the coupling constant, and the hopping term
	$T_\phi$ is a self-adjoint Toeplitz operator satisfying, for some $\rho>0$,
	\begin{equation}\label{lr}
		\phi_0=0,\quad	\phi_{-j}=\overline{\phi_j},\quad
		|\phi_j|\leq e^{-\rho|j|}, \quad  j\in \Z.
	\end{equation}
	In particular, choosing $\varepsilon\neq 0$,  $\phi_{\pm 1}=e^{-\rho }$ and $\phi_j=0$ for $|j|\ne1$, \eqref{md} gives the standard Schr\"odinger operators (up to a scalar).
		 Let $p_k/q_k$
	be the continued-fraction approximants of $\omega$, and set
	\begin{equation}\label{bet}
		\beta_k:=\frac{\log q_{k+1}}{q_k},\quad
		\beta:=\limsup_{k\to\infty}\beta_k.
	\end{equation}
	The quantity $\beta$ is also called  the exponential strength of frequency resonances, which  measures the exponential accuracy of the returns of the rotation by $\omega$.  Gordon-type  results  show that  for large $\beta$, these returns lead to infinitely many exponentially close repetitions of the potential that may  obstruct Anderson localization  (pure point spectrum with exponentially decaying
	eigenfunctions)  \cite{Gor76,AS82,Dam00,AYZ17,JK24b}. 

	 Arithmetic localization, which concerns how Anderson
	localization  depends on the arithmetic properties of the frequency
	and phase, such as the quality $\beta$, has attracted considerable interest in the study of quasiperiodic operators.

 Early studies of arithmetic localization focused on quasiperiodic Schr\"odinger operators with Diophantine frequencies (which satisfy $\beta=0$), where positivity of the Lyapunov exponent plays a  role in establishing Anderson localization. The nonperturbative strategy of exploiting the positivity of the Lyapunov exponent to prove Anderson localization was first developed   in \cite{Jit94,Jit99} for the almost Mathieu operator. This approach was subsequently extended to various settings \cite{BJ02,JK19,Kac19} and was also developed in a non-arithmetic way \cite{BG00,BJ00,BGS02,Kle05}; see also \cite{Bou05} and the references therein.  For Liouville frequencies (which satisfy $\beta>0$),   there is also  a common figure  that Anderson  localization for quasiperiodic Schr\"odinger operators depends on a competition between the Lyapunov exponent $L(E)$ and  $\beta$, where $L(E)>\beta$ is conjectured to be the universal  sharp arithmetic (localization) transition \cite{JK24}. 
 Sharp arithmetic localization was  originally   established  for the  almost Mathieu operator \cite{AYZ17,JL18}  and  Maryland model \cite{JL17}; see also \cite{AJ09,LY15,GYZ22,HJY22}. Recently,   the universality of the sharp arithmetic transition  has been established for   more general  quasiperiodic Schr\"odinger operators with   type $1$ potentials   \cite{GJ24,Han24}    and   Lipschitz monotone potentials   \cite{JK24}; see also  the very recent progress \cite{JLS26} on   potentials beyond uniform transversality.

However, arithmetic localization for long-range operators remains much  less understood  since the standard  Lyapunov exponent is not  available.   
It is also  worth noting that the arithmetic effects of $\omega$ on point spectrum  could  be quite different  in the   long-range setting \cite{YZZ14}.
This raises a natural question: what quantity can play an analogous role  of the Lyapunov exponent in the long-range setting, and whether a similar (not necessarily sharp)  competition with $\beta$ can lead to arithmetic localization? For long-range operators with tangent or cosine potentials,  several progress on arithmetic localization has been made in  \cite{FP84,BJ02,AJ10,YZ13,GJYZ26}; see also \cite{Ge24,HS26} for the finite-range case.  However, both techniques are based   either on fine trigonometric structure or  duality argument, relying on the  specific form of $v$.   This motivates the search for a long-range  arithmetic localization argument that relies on  qualitative properties of $v$, such as monotonicity, rather than  its special  structure.  
 Here we study  this question for long-range operators \eqref{md} with  $\alpha$-H\"older monotone potentials.  More precisely, the potential   $v$   is $1$-periodic 
 \begin{equation}\label{peri}
 	v: \R\to [-\infty,+\infty), \quad v(\theta+1)=v(\theta)
 \end{equation}
 and   $\alpha$-H\"older monotone on $[0,1)$ for some $\alpha\geq 1$ 
 \begin{equation}\label{HC}
 	v(\theta_2)-v(\theta_1)\geq(\theta_2-\theta_1)^\alpha,
 	\quad 0\leq \theta_1\leq\theta_2<1.
 \end{equation}  For  $\alpha=1$, $v$ is  Lipschitz monotone as considered in \cite{JK19,Kac19,JK24}.  	Note that if $v(0)=-\infty$ (equivalently, $v$ is unbounded below), \eqref{md} only defines operators for $\theta\in \R\setminus(\Z+\omega\Z)$. 

The goal of this paper is to establish   Anderson  localization in the regime $\rho >\alpha \beta$ for such monotone long-range  operators with  small couplings depending on an explicit bound  in terms of $\alpha$ and the gap $\rho-\alpha\beta$.
  Our main result is as follows.
	\begin{thm}\label{main}
		Suppose the operator family   \eqref{md} satisfies \eqref{lr}--\eqref{HC}  with    $\rho >\alpha \beta$.  Set $$
		\delta:=\rho-\alpha\beta>0.$$ Then  for any  $\varepsilon$ satisfying 
\begin{equation}\label{small1}
			0\leq \varepsilon \leq \varepsilon_0(\delta ,\alpha):= \frac12 (\frac{\delta}{100\alpha})^{2\alpha} \cdot (e^{\frac{1}{4}\delta}-1),
\end{equation}   $H(\theta)$ exhibits uniform Anderson localization for any $\theta$ where the operator is defined.  That is,  $H(\theta)$ has pure point spectrum, and  there exists $C>0$, uniformly in $\theta$ and $E$, such that for any eigenfunction $\psi$ satisfying $H(\theta)\psi=E\psi$, $\|\psi\|_{\ell^{\infty}}=1$, we have\begin{equation}\label{unif}
			|\psi(x)|\leq C e^{-\frac17\delta |x-x_\psi|}, \quad \forall x\in \Z, 
		\end{equation} where $x_\psi\in \Z $ satisfies $|\psi (x_\psi )|=\|\psi\|_{\ell^{\infty}}=1$. 
	\end{thm}
		\begin{rem}
		In \cite{Kac19,JK24}, the  integrability condition
		$$
		\int_0^1 \log(1+|v(x)|)dx<+\infty
		$$
		is needed  for the Lyapunov exponent to be  finite.
		The  present argument does not require this  integrability condition. 
	\end{rem}
		\begin{rem}
		Here  $\rho$ is similar in spirit to  the Lyapunov exponent $L(E)$, as it represents the expected exponential decay rate of the Green's function in the absence of resonances while frequency resonances may  lead to a loss of size $\alpha\beta$. However, the present analysis for  long-range operators produces an additional loss arising from the perturbation estimates for finite-volume eigenvalues.  Condition \eqref{small1} on $\varepsilon$  ensures that this additional loss is smaller than the gap $\rho-\alpha\beta$. 
	\end{rem}
	\begin{rem}
		While we do not claim that either the arithmetic condition $\rho>\alpha\beta$ or the coupling  bound $\varepsilon_0(\alpha,\delta)$ is optimal, our result provides a  readily verifiable sufficient condition for Anderson localization. It would be interesting to further investigate whether a sharp arithmetic transition exists for this class of long-range operators.
	\end{rem}

Let us  review  some of the results  in the literature related to this paper.
	The study of quasiperiodic Schr\"odinger operators with monotone potentials is motivated by 	 the  Maryland model. 
	   Classical results \cite{GFP82,FP84,Sim85} establish all couplings  localization for Maryland model for  Diophantine frequencies. A complete arithmetic transition of Maryland model  was  obtained in  \cite{JL17}; see also \cite{JY21,HJY22}.  However,  the techniques for Maryland model rely  essentially  on the particular structure of the tangent potential, which are not directly applied to potentials satisfying only a monotonicity condition.  To move beyond this special structure, Jitomirskaya and Kachkovskiy \cite{JK19} developed a nonperturbative approach that traces back to  \cite{Jit94}, for Schr\"odinger operators with bounded Lipschitz monotone potentials and established all couplings localization for Diophantine frequencies. This framework was subsequently extended to unbounded monotone potentials in \cite{Kac19}, thereby covering the Maryland model. A key feature of the arguments in \cite{JK19,Kac19} is that they reduce the localization problem to the positivity of the Lyapunov exponent $L(E)$, which is known to hold for almost every energy in this class of Schr\"odinger operators \cite{SS89,DK05}. The result of \cite{JK19} was later extended to certain Liouville frequencies in \cite{KZ24}, under a condition of the form $\beta<C L(E)$.  Combining the sharp treatment of frequency resonances in \cite{JL18},  \cite{JK24} established  the  sharp arithmetic localization in  the optimal regime $L(E)>\beta$. Together with the complementary Gordon-type result in \cite{JK24b}, this yields a sharp arithmetic  transition for Schr\"odinger operators with Lipschitz monotone potentials. See also   \cite{JLS26} for more recent progress on  potentials beyond  uniform transversality. 
We also  mention that Anderson  localization can be established by perturbative methods such as  KAM-type diagonalization  and   multi-scale analysis. One advantage of perturbative methods is that they extend directly  to (exponentially decaying) long-range operators and  higher-dimensional lattices $\Z^d$ in most cases, while  they  only cover Diophantine frequencies essentially  and the admissible size of the coupling  depends strongly   on  the   constants in Diophantine condition  of $\omega$.   Starting with the work of \cite{BLS83} on meromorphic monotone potentials, perturbative Anderson localization has been  extended to more general monotone potentials in \cite{KPS22,CSZ25,KPS25} and  more slowly decaying hopping  in \cite{Shi23,SW23}. 

	To the best of our knowledge, existing localization results for long-range operators with monotone potentials, apart from meromorphic ones,  have so far been obtained by perturbative methods. On the other hand, nonperturbative methods for long-range operators have been developed for certain analytic/meromorphic potentials in \cite{BJ02,Bou05,JSY21} based either  on specific trigonometric
	structure  or subharmonicity of the logarithms of finite-volume determinants.
   In this paper, we extend the nonperturbative localization methods developed for Schr\"odinger operators with monotone potentials in \cite{JK19,Kac19,JK24} to the long-range setting, in the spirit of \cite{BJ02}. We outline below the key new arguments   in this extension, compared to \cite{JK24,BJ02}.  
	\begin{itemize}
		\item[(1)] \textbf{Global $\alpha$-H\"older monotonicity via Cauchy interlacing argument.} To establish the large deviation theorem, we begin with a finite-volume eigenvalues analysis as in  \cite{JK24}. The first step is to construct reordered monotone eigenvalue branches $\Lambda_k$.  While local Lipschitz monotonicity implies global Lipschitz monotonicity, the same does not hold for H\"older monotonicity;  see also  \cite[Remark 4.3]{KPS25}. However,  here we use  a Cauchy interlacing argument to establish directly the required global $\alpha$-H\"older monotonicity  (cf. Proposition \ref{mon}).	 Moreover, this  argument shows that whenever the potential satisfies a quantitative  monotonicity with a given modulus, the reordered monotone  eigenvalue branches inherit the same monotonicity globally.
	\item[(2)] \textbf{Uniform eigenvalue counting difference via hopping truncation.}	For Schr\"odinger operators, cyclic permutations of the	finite-volume Laplacian produce only finite-rank perturbations. This observation  is crucial for establishing uniform bounds on the differences between eigenvalue counting functions at different phases in \cite{JK24}. 
	However, for long-range hopping, the corresponding difference is no	longer of uniformly bounded rank.	To deal with this, we truncate the hopping at a distance  $O(\log q)$, where $q$ is a continued-fraction
		denominator of $\omega$. The truncated hopping operator
		changes by a perturbation of rank $O(\log q)$ under cyclic
		permutations, while the truncation error has small norm of order $q^{-C}$, which  can be controlled using	the $\alpha$-H\"older monotonicity of $v$ and eigenvalue
		perturbation estimates. In this way, for a scale $n=sq+r$, we obtain a uniform in phase eigenvalue counting difference  $O(\log q+s+|r|)$ (cf. Lemma \ref{111}). This estimate is sufficient for the	subsequent large deviation argument.
		\item[(3)] \textbf{Large deviation theorem  without Lyapunov exponent.} For Schr\"odinger operators, Cramer's rule expresses the Green's function as a ratio of determinants restricted to integer intervals. These determinants can be estimated through products of finite-volume eigenvalue factors, with the main term estimated through the Lyapunov exponent in the large deviation theorem  \cite{JK24}. In the long-range setting, neither this simple determinant representation for numerator nor the Lyapunov exponent is available.  Here  we establish large deviation theorem directly for the  Green's function via  estimating  the numerator and denominator in Cramer's rule separately in the spirit of \cite{BJ02}. 
The key is to obtain compatible estimates for the numerator and denominator in Cramer's rule, so that their leading contributions cancel in the ratio. As in  \cite{Kac19,JK24}, we divide the finite-volume eigenvalues into small, middle, and large ones according to their distances from the energy $E$. For the denominator, we separate the contributions of the small and large eigenvalues as tail terms to be handled later. To handle  the middle eigenvalues, we follow the idea of \cite{Kac19} to  use the eigenvalue counting difference estimates above to compare their distribution with the values of the principal eigenvalue branch $\Lambda_0$ along the quasiperiodic orbit (cf. Proposition \ref{chabiel}). This allows us to express their main contribution in terms of a truncated logarithmic integral associated with the potential, up to a controlled error (cf. Remark \ref{1023}).
For the numerator, we use a  cofactor path expansion lemma  of \cite{BJ02} (cf. Lemma \ref{cof}), which reduces the problem to estimating determinants of restrictions to arbitrary subsets. By  Hadamard's inequality, we obtain upper bounds for these determinants, which   match the  contributions from the middle and large eigenvalues in the denominator (cf. Proposition \ref{Had}). These contributions cancel in Cramer's rule, up to an exponential error depending on $\varepsilon$.
The resulting Green's function estimate involves exponential decay at rate  $\rho$ from the hopping and path expansion, a negligible exponential error when    $\varepsilon$ is sufficiently small,  and  the product of the small eigenvalue factors remaining in the denominator.
 Thus, we obtain a large deviation theorem (cf. Theorem \ref{LDT}) similar to that  in \cite{JK19,Kac19,JK24},    showing that significant deviations from the expected exponential decay for  Green's function  are caused by finite-volume eigenvalues abnormally close to $E$. 
\item[(4)]	\textbf{Long-range Poisson formula iteration.} 
With  the large deviation theorem, we construct certain regular points and good integer  intervals (cf. Proposition \ref{regu}) by the  interpolation argument of \cite{JK24}. Here the concepts of certain regular point and good integer interval (cf. Definitions \ref{rd} and \ref{gq}) in our proof are  slightly different from that in \cite{JK24}, but the concepts resonant/nonresonant regions are similar.  The remaining  issue is to extend the  streamlined localization argument  of \cite{JK24} (see also \cite{JL18}), which uses different regularity estimates in resonant and nonresonant regions, to the long-range setting.  Unlike the Schr\"odinger case, the long-range Poisson formula allows jumps of arbitrary length. The key  is to exploit the exponential decay of the hopping coefficients to compensate for these jumps and control their accumulated contributions. 
By iterating  the Poisson formula on good intervals (the points in resonant regions need an  additional compensation argument),  we establish a one-step exponential decay estimate for the value of  generalized eigenfunctions at each resonant or nonresonant  point (cf. Proposition \ref{21151}). The decay rate depends on the point but admits a uniform positive lower bound corresponding to the worst case in the resonant region (cf. Remark \ref{rnm}). Finally, iterating this estimate until the iteration leaves certain annulus yields exponential decay of generalized eigenfunctions.
	\end{itemize}

\noindent	\textbf{Organization of the paper.} \S \ref{se2} collects the preliminary results on continued fractions, eigenvalue perturbation theory, and Cramer's rule. In \S \ref{Rellich}, we study the monotonicity and distribution of finite-volume eigenvalues,  which is the key  ingredient for the subsequent determinant estimates and interpolation argument.
	\S	 \ref{deno} and \S \ref{num} are devoted to estimating the denominator and numerator in Cramer's rule, respectively.  Combining these estimates, we establish in \S \ref{ldt} a large deviation theorem for the Green's function.  We also specify the smallness condition on the coupling constant needed to preserve the  exponential decay.  
	Finally, in \S \ref{al}, we use the large deviation theorem and interpolation argument to construct good intervals in resonant and nonresonant regions. By a  Poisson formula iteration on these good intervals,   we prove Theorem \ref{main}.

	\section{Preliminaries}\label{se2}
	In this section, we collect some  standard facts  that will be used throughout the paper.  
	\subsection{Continued fraction and discrepancy estimate}\label{pre1} We refer to  \cite{KN74} for the  results in this subsection.

	Let	$
	\omega\in \mathbb R\setminus\mathbb Q
	$
	and write its continued-fraction expansion as
	\[
	\omega=[a_0;a_1,a_2,\dots]
	=
	a_0+\frac{1}{
		a_1+\frac{1}{
			a_2+\cdots
	}}.
	\]
	Denote the $k$-th continued-fraction approximation of $\omega$ by
	\[
	\frac{p_k}{q_k}
	=
	[a_0;a_1,a_2,\dots,a_k]=	a_0+\frac{1}{
		a_1+\frac{1}{
			a_2+\frac{1}{\cdots+a_k}
	}}.
	\]
	\begin{prop}[Continued-fraction approximations estimate]
		The denominator  $q_k$ satisfies  the following  good approximation  property 
		\begin{align}
			|
			\omega-\frac{p_k}{q_k}
			|&\in 
			(\frac{1}{2q_kq_{k+1}}, \frac{1}{q_kq_{k+1}}),\label{good}\\
			\min_{1\leq j<q_{k+1}}\|j\omega\|_\T=  	\|q_k\omega\|_{ \T}	&\in   (	\frac{1}{2q_{k+1}}, 	\frac1{q_{k+1}}),\label{best}
		\end{align}
		where
		$
		\|\cdot\|_{\mathbb T}:
		=
		\min_{m\in\mathbb Z}|\cdot-m|$.
	\end{prop}

	We also need the following Denjoy-Koksma inequality from  discrepancy theory.
	\begin{prop}[Denjoy-Koksma inequality] \label{dj}
		Let $$f\in {\rm BV}(\mathbb{T})=\{f  \text{ is $1$-periodic}:\  \operatorname{Var}_{\mathbb{T}}(f)<+\infty \},$$ where
		\[
		\operatorname{Var}_{\mathbb{T}}(f)
		:=
		\sup
		\{
		\sum_{i=0}^{N-1}
		\left|
		f(\theta_{i+1})-f(\theta_i)
		\right|
		:
		\theta_0<\theta_1<\cdots<\theta_{N-1}<\theta_N<\theta_0+1
		\}.
		\]
		Let $q_k$ be  the
		denominator of the $k$-th continued-fraction approximation of $\omega$ and let $n=sq_k+r$ ($s\geq 1$, $|r|\leq q_k-1$).  Then   for any $\theta\in\T$,
		$$\left|
		\sum_{j=0}^{n-1}f(\theta+j\omega)
		-
		n\int_0^1 f(\theta)\,d\theta
		\right|
		\leq
		(s+|r|)\operatorname{Var}_{\mathbb{T}}(f).$$	\end{prop} 
	Applying  the above inequality with $f$ being the indicator function of an interval  $I$ on $\T$,  we obtain the following proposition. 
	\begin{prop}\label{DJC}
		Let $q_k$ be  the
		denominator of the $k$-th continued-fraction approximation of $\omega$ and let $n=sq_k+r$ ($s\geq 1$, $|r|\leq q_k-1$).  Let $I\subset\mathbb{T}$ be an interval and denote by
		$|I|$ its Lebesgue measure.  Then   for any $\theta\in\mathbb{T}$, 
		$$
		\left|
		\#\left\{
		0\leq j\leq n-1:
		(\theta+j\omega)  \bmod1  \in I
		\right\}
		-
		n|I|
		\right|
		\leq
		2(s+|r|).
		$$
	\end{prop}
	In the rest of this paper,  whenever $q$ appears, we assume
	that $q=q_k$ is the  denominator of a continued-fraction approximation of
	$\omega$. Moreover, whenever we write $n=sq+r$, we always assume that
	$s\geq 1$ and $|r|\leq q-1$.

	\subsection{Eigenvalue perturbation theory for  self-adjoint operators}\label{SB} We refer to  \cite{Kat95} for the  results in this subsection. 
	
	For  an $n$-dimensional self-adjoint operator (matrix) $A$,    we denote its operator norm by $\|A\|$ and denote its 
	eigenvalues in non-decreasing order by 
	$$\lambda_0(A)\leq \lambda_1(A)\leq \cdots\leq \lambda_{n-1}(A). $$
	We first recall the   standard min-max principle.
	\begin{lem}[Min-max principle]
		\label{minmax}
		Let $A$   be an $n$-dimensional self-adjoint operator.  Then  for every
		$0\leq k\leq n-1$,
		$$
		\lambda_k(A)
		=
		\min_{\substack{V\subset \mathbb{C}^n\\ \dim V=k+1}}
		\;
		\max_{\substack{u\in V\\ \|u\|=1}}
		\langle u,Au\rangle .
		$$
	\end{lem}
	The min-max principle immediately gives the following monotonicity/stability estimate
	for eigenvalues under perturbations.  
	\begin{prop}[Eigenvalue monotonicity]
		\label{mot}
		Let $A$ and $B$ be two $n$-dimensional self-adjoint operators and $a\in \R$ satisfying $$A+ a\operatorname{Id}_n\leq B,$$ where $\operatorname{Id}_n$ is the $n$-dimensional identity operator and the ``$\leq $'' between two operators  is in the sense of quadratic form.  Then  for every
		$0\leq k\leq n-1$,
		$$
		\lambda_k(A)+a\leq \lambda_k(B). 
		$$
	\end{prop}
	\begin{prop}[Eigenvalue stability]
		\label{sta}
		Let $A$ and $B$ be two $n$-dimensional self-adjoint operators. Then  for every
		$0\leq k\leq n-1$,
		\[
		|\lambda_k(A)-\lambda_k(B)|
		\leq \|A-B\|.
		\]
	\end{prop}
	We will also use the following consequence of the min-max
	principle for finite-rank perturbations.
	
	\begin{prop}[Finite-rank perturbation]
		\label{pert}
		Let $A$ and $B$ be two $n$-dimensional self-adjoint operators and suppose that
		$$
		\rank(A-B)= r.
		$$
		Then  for every $0\leq k\leq n-1-r$,
		$$
		\lambda_k(A)\leq \lambda_{k+r}(B),
		\quad
		\lambda_k(B)\leq \lambda_{k+r}(A).
		$$
	\end{prop}
	Finally, we recall  the Cauchy interlacing theorem, which is also a standard consequence of the min-max principle. 
	\begin{prop}[Cauchy interlacing theorem] 
		\label{Cau}
		Let $A$ be an $n$-dimensional self-adjoint operator, and let $P$ be an
		orthogonal projection with
		$$
		\dim(\Ran P)=n-r.
		$$
		Regard the $(n-r)$-dimensional compressed operator $PAP$ as a self-adjoint operator on $\Ran P$.
		Then  for every $0\leq k\leq n-r-1$,
		\[
		\lambda_k(A)
		\leq
		\lambda_k\bigl(PAP|_{\Ran P}\bigr)
		\leq
		\lambda_{k+r}(A).
		\] 
	\end{prop}
	\subsection{Green's function, Cramer's rule}
	For a finite set $\Gamma \subset\Z$, we denote by $R_\Gamma $ the restriction operator on $\Gamma $ and  $H_\Gamma (\theta):= R_\Gamma H(\theta)R_\Gamma $ the Dirichlet restriction of $H(\theta)$ on  $\Gamma $. For a set $\Lambda\subset\Z$ and a point $x\in \Z$, denote 
	$$\Gamma+x:=\{y+x:\ y\in \Gamma \}. $$
	Then one has \begin{equation}\label{mosa}
		H_{\Gamma+x} (\theta)=H_{\Gamma} (\theta+x\omega).
	\end{equation}
	For $E\in\mathbb C$, define the Green function, when it exists, by  
	$$
	G_\Gamma^{\theta,E}:=(H_\Gamma(\theta)-E)^{-1}.
	$$
For $a\leq b$, we denote the integer interval $[a;b]:=[a,b]\cap \Z$.  
	For $n\geq1$,  we let  $$I_n=[0;n-1] ,\quad  H_n(\theta):=H_{I_n}(\theta), \quad P_n(\theta,E)=\operatorname{det}(H_n(\theta)-E). $$
	By Cramer's rule, for $x,y\in I_n$, 
	\begin{equation}\label{cra}
		G_{I_n}^{\theta,E}(x,y)=\frac{\operatorname{cof}_{y,x}(H_{n}(\theta)-E)}{P_n(\theta,E)},
	\end{equation} where  $\operatorname{cof}_{y,x}(\cdot)$ is the $(y,x)$-cofactor of a matrix.
	The  aim of \S \ref{Rellich}--\ref{ldt}  is to obtain a  large deviation theorem for the Green's function (cf. Theorem \ref{LDT}),  which is  crucial for  the proof  of Anderson localization. 
	
	\section{Finite volume eigenvalues analysis}\label{Rellich} 
	In this section, we  study the phase dependence and the
	distribution of the finite-volume eigenvalues of $H_n(\theta)$, which is an ingredient of the  large deviation theorem.    Let 
	$$E_0(\theta)\le\cdots\le E_{n-1}(\theta)
	$$ be the non-decreasing ordered eigenvalues of $H_n(\theta)$. (With  the notation in \S \ref{SB}, $E_k(\theta)=\lambda_k(H_n(\theta))$.) 
	
	Note that the diagonal entries of  $H_n(\theta)$ are $\{v(\theta+k\omega)\}_{k=0}^{n-1}$. Thus one of  the diagonal entries of  $H_n(\theta)$    has a  negative jump at   
	$$ 0=\vartheta_0<\vartheta_1<\cdots<\vartheta_{n-1}<\vartheta_n=1,\quad  
	\{\vartheta_0,\cdots,\vartheta_{n-1}\}=
	\{\{-l\omega\} :\ 0\le l\leq n-1\},
	$$
	where $$\{\theta\}:=\theta\bmod 1 \in [0,1).$$ 
	We call the $\vartheta_j$ jump points. Between two consecutive jump
	points, no diagonal entry crosses an integer, where condition \eqref{HC}
	applies to every diagonal entry so that 
	$$H_n(\theta_2)\geq H_n(\theta_1)+ (\theta_2-\theta_1)^\alpha\operatorname{Id}_n , \quad   \vartheta_{j}<\theta_1\leq \theta_2< \vartheta_{j+1}. $$ 
	The above inequality and  Proposition \ref{mot} give,
	for $0\leq k,j<n$,
	$$E_k(\theta_2)-E_k(\theta_1)\geq (\theta_2-\theta_1)^\alpha , \quad   \vartheta_{j}<\theta_1\leq \theta_2< \vartheta_{j+1}. $$ 
	However, $E_k(\theta)$ need not be  globally
	monotone since 	at a jump point, one diagonal entry returns to the start of its
	period, which causes a negative rank one perturbation. However,  as in \cite{JK19,Kac19,JK24}, one can  renumber them across the jump points so that the resulting branches retain the monotonicity.
	\subsection{$\alpha$-H\"older monotonicity of reordered eigenvalue branches}
	For $m\in \Z$, denote $$[m]_n:=m\bmod n \in [0;n-1].$$	Define the reordered eigenvalue branches  $\Lambda_k(\theta)$ $(0\leq k\leq n-1) $ as   
	\begin{equation}\label{2151}
		\Lambda_k(\theta )
		: =
		E_{[k+ j]_n}(\theta ),
		\quad
		\theta \in[\vartheta_ j,\vartheta_{ j+1}), \quad0\leq k,j\leq n-1
	\end{equation}
	on $[0,1)$   and extended  $1$-periodically.  (If $v(0)=-\infty$,  we replace the above  interval ``$[\vartheta_ j,\vartheta_{ j+1})$'' by $(\vartheta_ j,\vartheta_{ j+1})$; see also Remark \ref{ub} bellow.)       
	Since the collection $\{\Lambda_k(\theta):0\le k\leq n-1\}$ 
	is merely a re-numbering of the finite-volume eigenvalues $\{E_k(\theta):0\le k\leq n-1\}$,  we have  \begin{equation}\label{pu}
		\operatorname{Spec}(H_n(\theta))=\{E_0(\theta),\cdots,E_{n-1}(\theta)\} =\{\Lambda_0(\theta),\cdots,\Lambda_{n-1}(\theta)\},
	\end{equation} where $\operatorname{Spec}(\cdot)$ is the  spectrum (eigenvalues)  counted in multiplicity of an operator.  
	
	The following  proposition gives  $\alpha$-H\"older monotonicity for $\Lambda_k(\theta)$ on an appropriate periodic interval. 
	\begin{prop}\label{mon}
		For each $0\leq k\leq n-1$, 	$\Lambda_k(\theta)$ is   $\alpha$-H\"older monotone on $[\vartheta_{n-k},\vartheta_{n-k}+1)$. More precisely, for $\theta_1,\theta_2$ belong to the phases where $H_n(\theta)$ is  defined, 
		\begin{equation}	\label{Hol}
			\Lambda_k(\theta_2)-\Lambda_k(\theta_1)
			\geq
			(\theta_2-\theta_1)^\alpha,
			\quad
			\vartheta_{n-k}\leq \theta_1\leq \theta_2<\vartheta_{n-k}+1. 
		\end{equation}

	\end{prop}
	
	\begin{proof}
		Without loss of generality, we give the proof for  $k=0$ since the other cases are similar by translation argument and periodicity. By periodicity, $[\vartheta_n,\vartheta_n+1)=[1,2)$ is identified with $[0,1)$.  Let $0\leq \theta_1\leq \theta_2<1$,  and suppose that $\theta_1\in [\vartheta_{j_1},\vartheta_{ j_1+1})$ and $\theta_2\in [\vartheta_{j_2},\vartheta_{ j_2+1})$ so that  there are  exactly \(r=j_2-j_1\)  many jump points
		lie in \((\theta_1,\theta_2]\).
		By the definition of $\Lambda_0(\theta)$  (cf. \eqref{2151}), 
		\begin{equation}\label{1050}
			\Lambda_0(\theta_1)=E_{j_1}(\theta_1), \quad \Lambda_0(\theta_2)=E_{j_2}(\theta_2).
		\end{equation}
		Let \(J\subset\{0,\ldots,n-1\}\) be the corresponding
		\(r\) diagonal entries which cross an integer in \((\theta_1,\theta_2]\), and let
		\(R_J^\perp\) denote the orthogonal projection onto their complement. 	Let $\mu _0(\theta)\le\cdots\le\mu _{n-r-1}(\theta)$ be the non-decreasing ordered eigenvalues of the 
		\((n-r)\)-dimensional compressed operator 	$R_J^\perp H_{n}(\theta)R_J^\perp$.  		Cauchy interlacing theorem  (cf. Proposition \ref{Cau})  gives
		\begin{equation}\label{1049}
			E_{j_1}(\theta_1)\leq \mu_{j_1}{(\theta_1)},
			\quad
			\mu _{j_1}(\theta_2)\leq E_{j_1+r}(\theta_2)= E_{j_2}(\theta_2),  
		\end{equation}
		For every \(j\notin J\),
		\[
		v(\theta_2+j\omega)-v(\theta_1+j\omega)\ge	(\theta_2-\theta_1)^\alpha.
		\]
		Since the hopping is independent of the phase,
		\begin{equation*}
			R_J^\perp H_{n}(\theta_2)R_J^\perp
			\ge
			R_J^\perp H_{n}(\theta_1)R_J^\perp
			+
			(\theta_2-\theta_1)^\alpha \operatorname{Id}_{n-r}. 
		\end{equation*}
		The above inequality  and Proposition \ref{mot} give 
		\begin{equation}\label{key2}
			\mu_{j_1}(\theta_2)
			\ge
			\mu_{j_1}(\theta_1)+(\theta_2-\theta_1)^\alpha.
		\end{equation}
		Combining \eqref{1050}, \eqref{1049} and \eqref{key2} yields 
		\[
		\Lambda_0(\theta_2)
		\ge
		\Lambda_0(\theta_1)+(\theta_2-\theta_1)^\alpha.
		\]
	\end{proof}
	\begin{rem}\label{ub}
		In the case   $v(0)=-\infty$,   we can extend the definition of $\Lambda_k$  to jump point $\vartheta_j=\{-l_j\omega \}$  by  
		$$
		\Lambda_k(\vartheta_j): =
		\lambda_{[k+ j ]_n -1}(R_{\{l_j\}}^\perp H_n(\vartheta_j)R_{\{l_j\}}^\perp),
		\quad 0\leq k,j\leq n-1,$$  with  $\Lambda_k(\vartheta_{n-k})=\lambda_{ -1}(R_{\{l_{n-k}\}}^\perp H_n(\vartheta_{n-k})R_{\{l_{n-k}\}}^\perp)$ understood by $-\infty$.	 With this extension, $\Lambda_k:\ \R\to [-\infty,+\infty)$  is $1$-periodic and $\alpha$-H\"older monotone on $[\vartheta_{n-k},\vartheta_{n-k}+1)$ by Cauchy interlacing argument.   So in the rest of this paper, we can   view  $\Lambda_k(\theta )$ is defined on $\R$ for convenience. 
	\end{rem}
	\subsection{Quasi-zeros}\label{Rellichqz}
	By \eqref{pu}, we have 
	$$
	|P_n(\theta,E)|
	=\prod_{k=0}^{n-1}|E_k(\theta)-E|=\prod_{k=0}^{n-1}|\Lambda_k(\theta)-E|. 
	$$
	
	As in \cite{Kac19,JK24}, we divide the finite-volume eigenvalues into three
	types according to their distances from $E$. We call the eigenvalues
	close to $E$ \textbf{small eigenvalues}, those far away from $E$
	\textbf{large eigenvalues}, and the remaining ones \textbf{middle
		eigenvalues}, where the  ``close'' and ``far'' will be specified later. These three types of eigenvalues  will be treated separately  in the
	estimate of $|P_n(\theta,E)|$.
	
	For the analysis of  the small eigenvalues,  we define for  $0<b<1$, \begin{equation}\label{ps}
		P_n^{<b}(\theta,E)
		:=
		\prod_{\substack{0\leq k\leq n-1\\
				|E_k(\theta)-E|<b}}
		(E_k(\theta)-E),
	\end{equation}
	with convention that 	$P_n^{<b}(\theta,E)=1$ if  $\{0\leq k\leq n-1:\ 
	|E_k(\theta)-E|<b\}=\emptyset$. 
	In other words, $P_n^{<b}(\theta,E)$ is the product of all factors of $P_n(\theta,E)$ whose values are less than $b$.  
	
	To estimate $P_n^{<b}(\theta,E)$,  we follow the idea of  \cite{JK24} to   define  the following  quasi-zeros of the	monotone branches $\Lambda_k$. They play the same role as the
	intersection points therein.  
	\begin{defn}\label{qz}
		For  $E\in\mathbb R$ and a  branch $\Lambda_k(\theta )$, define
		\begin{equation}
			z_k(E)
			=
			\inf
			\left\{
			\theta \in [\vartheta_{n-k},\vartheta_{n-k}+1):\ 
			\Lambda_k(\theta )\geq E
			\right\},
		\end{equation}
		with convention that  $ 	z_k(E)= \vartheta_{n-k}+1$ if the set is empty. 
		We denote by
		\[
		\widetilde z_k(E)\in\mathbb T  \text{ 	the image of $z_k(E)$ modulo $ \Z$}
		\]
		and call it the
		\textbf{quasi-zero} of $\Lambda_k$ at energy $E$.
		Obviously, if the branch actually takes the value $E$, then
		$\widetilde z_k(E)$ is an actual zero of
		$\Lambda_k(\theta  )-E$.

	\end{defn}
	The monotonicity of $\Lambda_k$
	allows us to bound $|\Lambda_k(\theta)-E|$ from below by the distance
	from $\theta$ to the corresponding quasi-zero. In this way, $P_n^{<b}(\theta,E)$ can be estimated in terms of the distances from $\theta$ to these quasi-zeros. 
	\begin{prop}
		For any  $\theta\in\mathbb T$, we have 
		\begin{equation}\label{chengj}
			|
			\Lambda_k(\theta)-E
			|
			\geq
			\|
			\theta-\widetilde z_k(E)
			\|_{\T}^{\alpha}.
		\end{equation}
		As a consequence, 
		\begin{equation}\label{sf}
			|P_n^{<b}(\theta,E)|
			\geq
			\prod_{\substack{0\le k\leq n-1\\
					\|\theta-\widetilde z_k(E)\|_{\mathbb T}<b^{1/\alpha}}}
			\|\theta-\widetilde z_k(E)\|_{\mathbb T}^{\alpha},
		\end{equation}
		with convention that the right-hand side of \eqref{sf} equals to $1$  if  $\{0\leq k\leq n-1:\ 
		\|\theta-\widetilde z_k(E)\|_{\mathbb T}<b^{1/\alpha} \}=\emptyset$. 
	\end{prop}
	\begin{proof}
		To prove \eqref{chengj}, 	we first work on the interval $[\vartheta_{n-k},\vartheta_{n-k}+1)$  and  $z_k(E)$.
		Suppose $\theta>z_k(E)$.  By the definition of $z_k(E)$, one may choose a sequence 
		$\theta_i \to  z_k(E)+$ such that
		\[
		\Lambda_k(\theta_i)\ge E. 
		\]
		Using \eqref{Hol}, we have 
		\[
		\Lambda_k(\theta)-E
		\ge
		\Lambda_k(\theta)-\Lambda_k(\theta_i)
		\ge
		(\theta-\theta_i)^\alpha.
		\]
		Letting $i\to\infty$ gives
		\begin{equation}\label{2114}
			\Lambda_k(\theta)-E
			\ge
			(\theta-z_k(E))^\alpha.
		\end{equation}
		Similarly, suppose $\theta<z_k(E)$.  Choose $\theta_i \to  z_k(E)-$ such that
		\[
		\Lambda_k(\theta_i)\le E.
		\]
		Then  using \eqref{Hol}, we have 
		\[
		E-\Lambda_k(\theta)
		\ge
		\Lambda_k(\theta_i)-\Lambda_k(\theta)
		\ge
		(\theta_i-\theta )^\alpha.
		\]
		Letting $i\to\infty$ gives
		\begin{equation}\label{2115}
			E-\Lambda_k(\theta)
			\ge
			(z_k(E)-\theta) ^\alpha.
		\end{equation}
		Combining \eqref{2114} and \eqref{2115} yields 
		\begin{equation}
			\label{e2}
			|\Lambda_k(\theta)-E|
			\ge
			|\theta-z_k(E)|^\alpha,
			\quad
			\theta \in  [\vartheta_{n-k},\vartheta_{n-k}+1).
		\end{equation}
		For an arbitrary $\theta\in\mathbb T$, choose integers $m_1,m_2$
		so that
		\[
		\theta+m_1,\ \widetilde{z} _k(E)+m_2= z_k(E)    \in  [\vartheta_{n-k},\vartheta_{n-k}+1). 
		\]
		By the $1$-periodicity of $\Lambda_k$  and
		\eqref{e2},
		$$
		|\Lambda_k(\theta)-E|
		\geq
		|(\theta+m_1)-(\widetilde{z} _k(E)+m_2)|^\alpha\geq 	\|\theta-\widetilde z_k(E)\|^\alpha_{\T}.
		$$ This finishes the proof of \eqref{chengj}. 
		
		Next we prove \eqref{sf}.  By \eqref{pu} and \eqref{ps}, we have 
		\begin{equation}\label{jiaohuan}
			|P_n^{<b}(\theta,E)|
			=
			\prod_{\substack{0\leq k\leq n-1\\
					|\Lambda_k(\theta)-E|<b}}
			|\Lambda_k(\theta)-E|.
		\end{equation}
		By \eqref{chengj},
		\begin{equation}\label{duole}
			|\Lambda_k(\theta)-E|<b
			\quad\Longrightarrow\quad
			\|\theta-\widetilde z_k(E)\|_{\mathbb T}<
			b^{1/\alpha}.
		\end{equation}
		Putting \eqref{chengj}, \eqref{jiaohuan} and \eqref{duole} together, we obtain \eqref{sf}. 	Notice  that the product on the right may contain more factors than the
		product on the left.  This does not reverse the inequality, since every
		additional factor is smaller than $b<1$.
	\end{proof}
	\subsection{Eigenvalue counting function analysis}
	Denote by
	$$
	V_n(\theta)
	:=
	\operatorname{diag}
	\bigl(v(\theta),v(\theta+\omega),\ldots,
	v(\theta+(n-1)\omega)\bigr)
	$$
	the diagonal part of $H_n(\theta)$. By the monotonicity of $v$ on $[0,1)$,  
	\begin{equation}\label{dh}
		v(\theta)=\lambda_j(V_n(\theta)), \quad \theta\in[\vartheta_j,\vartheta_{j+1}).
	\end{equation}
	On the other hand, by the definition \eqref{2151} of the reordered
	eigenvalue branches,
	\begin{equation}\label{wd}
		\Lambda_0(\theta)=E_j(\theta)=\lambda_j(H_n(\theta)),
		\quad
		\theta\in[\vartheta_j,\vartheta_{j+1}).
	\end{equation}
	Thus $\Lambda_0(\theta)$ occupies the same position among the eigenvalues
	of $H_n(\theta)$ as $v(\theta)$ does among the eigenvalues of  $V_n(\theta)$. In this sense, $\Lambda_0$ behave similar to  the potential $v$. We therefore call $\Lambda_0$ the
	\textbf{principal branch}.
	Set \begin{equation}\label{633}
		A_\rho:=\frac{2}{e^\rho-1}.
	\end{equation} 	
	Combining \eqref{dh}, \eqref{wd} and  Proposition \ref{sta} yields 
	\begin{equation}\label{duibi}
		|v(\theta)-\Lambda_0(\theta)|
		\leq
		\|H_n(\theta)-V_n(\theta)\|
		\leq
		\varepsilon\|T_\phi\|
		\leq A_\rho\varepsilon,
	\end{equation}
	where  we used  \begin{equation}\label{566}
		\|T_\phi\|
		\leq \|\phi\|_{\ell^1}
		\leq A_\rho
	\end{equation} due to  Schur test and  \eqref{lr} in the last inequality of \eqref{duibi}.

	For $\theta\in\T $ and $E\in\R$, define the eigenvalue counting function
	\begin{equation*}
		N_n(\theta,E)
		:=
		\#	\{0\leq k\leq n-1:\ E_k(\theta)\leq E\},
	\end{equation*}
	and the counting function for the values of  $\Lambda_0$ sampled along the
	quasiperiodic orbit
	$\{\theta,\theta+\omega,\cdots,\theta+(n-1)\omega\}$
	\[
	N_n^{\Lambda_0}(\theta,E)
	:=
	\#\{0\leq k\leq n-1:\
	\Lambda_0(\theta+k\omega)\leq E\}.
	\]
	To estimate the  contribution of the middle eigenvalues to
	$P_n(\theta,E)$, we follow the idea of \cite{Kac19} to  compare 	$N_n(\theta,E)$ with
	$	N_n^{\Lambda_0}(\theta,E)$. 
	However, the argument for Schr\"odinger operators in \cite{Kac19}  relies on the fact
	that, after a cyclic permutation, the Laplacian changes only by a finite-rank perturbation, which does not hold for the long-range hopping operator. 
Hence, we need to establish new estimates suitable for the present long-range setting.

	Before stating the comparison estimate, let us explain why such a result
	is natural. If $\varepsilon=0$, then $H_n(\theta)=V_n(\theta)$ and by \eqref{duibi},  $
	\Lambda_0(\theta)=v(\theta).$
	Moreover,
	\[
	\{E_k(\theta):0\leq k\leq n-1\}
	=
	\{v(\theta+k\omega):0\leq k\leq n-1\}.
	\]
	Therefore, in this case,
	$$
	N_n(\theta,E)=N_n^{\Lambda_0}(\theta,E).
	$$
	When $\varepsilon\neq 0$,  this exact identity is no
	longer true. 
	The following proposition shows that the error is at most of order
	$\ln q+s+|r|$ for $n=sq+r$\footnote{As a reminder,  $q$ denotes the  denominator of a continued-fraction approximation
		of $\omega$ and whenever we write $n=sq+r$, we always assume that
		$s\geq 1$ and $|r|\leq q-1$.}.  In other words, even in the presence of the long-range
	perturbation, the eigenvalue distribution of $H_n(\theta)$ is still well
	approximated by 
	$\{\Lambda_0(\theta+k\omega)\}_{k=0}^{n-1}$.

	\begin{prop}\label{chabiel}	
		There exists  some $C_1=C_1(\alpha,\rho ,\varepsilon)>0$ such that  for $ 	n=sq+r$,
		\begin{align}	
			\sup_{\theta,E}
			|N_n(\theta,E)-N_n^{\Lambda_0}(\theta,E)| \leq \nonumber C_1 (\ln q+s+|r|) , 
		\end{align}
	\end{prop}
	Before starting  the proof of Proposition \ref{chabiel}, we  state its important consequence which will be used  to  estimate the  contribution of the middle eigenvalues to
	$P_n(\theta,E)$  in the next section.   
	\begin{cor}\label{BV}
		For every function $g$ of bounded variation on $\R$,
		\begin{equation*}
			\left|
			\sum_{k=0}^{n-1}g(E_k(\theta)-E)
			-
			\sum_{k=0}^{n-1}g(\Lambda_0(\theta+k\omega)-E)
			\right|
			\le
			C_1(\ln q+s+|r|)\operatorname{Var}(g),
		\end{equation*} where $g(-\infty):=\lim_{t\to-\infty}g(t)$. The limit  exists because  $g$ has bounded variation. 
	\end{cor}
	\begin{proof}[Proof of Corollary \ref{BV}]
		This follows from  Stieltjes integration by parts to the two counting functions. 
		Proposition \ref{chabiel} gives the upper bound for the difference of the two counting functions.
	\end{proof}
	The proof of Proposition~\ref{chabiel} requires some preparation. We	first prove Proposition~\ref{chabiel} assuming the following two lemmas,
	and then prove the two lemmas. The first lemma gives a
	uniform comparison between the eigenvalue counting functions at
	different phases, which can be seen as the long-range counterpart of \cite[Corollary 3.5]{JK24}. The second lemma is a discrepancy estimate, which
	follows from Proposition~\ref{DJC}.
	\begin{lem} \label{111}	For $ 	n=sq+r$, 
		we have 
		\begin{align}	\label{22}
			\sup_{\theta_1,\theta_2,E}
			|N_n(\theta_1,E)-N_n(\theta_2,E)|
			&\le
			100( \frac\alpha \rho \ln q     + (1+(\frac{2\varepsilon}{\rho })^{\frac 1\alpha})s+|r|+1).	
		\end{align}
	\end{lem} 
	\begin{lem} \label{l32}
		For $t\in\T$, let $l(t)\in[0;n-1]$  be the index of the  interval $[\vartheta_{l(t)}, \vartheta_{l(t)+1})$ containing $t$. Then  uniformly in $\theta $ and $M\in [0;n]$, we have 
		$$
		|
		\#\{0\le j<n: l(\theta +j\omega)<M\}-M
		|
		\le
		4(s+|r|).
		$$
	\end{lem}
	
	\begin{proof}[Proof of Proposition \ref{chabiel}]
		Set $M=N_n(\theta,E)$ and
		$m_j=N_n(\theta+j\omega,E) $.
		By Lemma \ref{111}, 
		\begin{equation}\label{k1y}
			|m_j-M|\le D,\quad D= 100( \frac\alpha \rho \ln q     + (1+(\frac{2\varepsilon}{\rho })^{\frac 1\alpha})s+|r|+1).
		\end{equation}
		By the definition of $\Lambda_0(\theta)$ (cf. \eqref{2151}),   $\Lambda_0(\theta+j\omega)=E_{l(\theta +j\omega)}(\theta +j\omega)$. Hence,  
		\[
		\Lambda_0(\theta+j\omega)\le E
		\quad\Longleftrightarrow\quad
		l(\theta +j\omega)<m_j.
		\]
		Thus \begin{equation}\label{k211}
			N_n^{\Lambda_0}(\theta,E)=\#\{0\leq j\leq n-1:\  l(\theta +j\omega)<m_j \}. 
		\end{equation}
		By \eqref{k1y} and \eqref{k211},  
		$$ \#\{0\leq j\leq n-1:\  l(\theta +j\omega)<M-D \}	\leq N_n^{\Lambda_0}(\theta,E) \leq \# \{0\leq j\leq n-1:\  l(\theta +j\omega)<M+D \}. $$
		Now we can combine  the above inequality and Lemma \ref{l32} to obtain $$ | N_n^{\Lambda_0}(\theta,E)- M|\leq D+ 	4(s+|r|)\leq  C_1 (\ln q+s+|r|).$$ This completes the proof of Proposition \ref{chabiel}. 
	\end{proof}
	In the rest of this section, we  prove Lemmas \ref{111} and \ref{l32}. 
	\begin{proof}[Proof of Lemma \ref{111}] 
		If $r\ne0$, compare $H_{n}(\theta)$ with its compressed operator $H_{sq}(\theta)$. Cauchy  interlacing theorem  gives
		\[
		|N_{n}(\theta,E)-N_{sq}(\theta,E)|\le|r|, \quad \forall \theta \in \T. 
		\]
		By the above inequality and triangle inequality, to prove Lemma \ref{111}, it suffices to prove  it for $n=sq$.  
		
		For  $R\in \Z^+$,  define the $R$-truncated hopping  operator $$T_\phi^{(R)}(i,j):=\begin{cases}
			T_\phi(i,j), \quad & |i-j|\leq R,\\
			0, \quad & |i-j|> R,
		\end{cases} \qquad  T_{\phi,n}^{(R)}=   R_{I_n} T_\phi^{(R)}  R_{I_n} $$
		and  
		$$  H^{(R)}_n(\theta) = V_n(\theta)+ \varepsilon T_{\phi,n}^{(R)}.$$
		In other words, $T_{\phi,n}^{(R)}$ and  $H^{(R)}_n(\theta)$ are the matrices obtained by deleting the off-diagonal entries of $R_{I_n} T_\phi  R_{I_n}$  and $H_n(\theta)$ outside the $R$-band.  Let $N_n^{(R)}(\theta,E)$ be the eigenvalue counting function of $H^{(R)}_n(\theta)$ 
		$$	N_n^{(R)}(\theta,E)
		:=
		\#	\{0\leq k\leq n-1:\ \lambda_k(H^{(R)}_n(\theta))\leq E\}.$$
		
		We will split the proof of \eqref{22} into two lemmas.  
		The first one is an analogue of \eqref{22} for  $N_n^{(R)}(\theta,E)$, using  the observation  that the  $R$-truncated hopping term $ T_{\phi,n}^{(R)}$,  after shift  unitary  conjugations,  only brings  rank $O(R)$ perturbation. 
		\begin{lem}\label{1828}
			For $n=sq$,  we have  \begin{equation}
				\sup_{\theta_1,\theta_2,E}
				|N_n^{(R)}(\theta_1,E)-N_n^{(R)}(\theta_2,E)|
				\le 7 (R+s+1).
				\label{1220}
			\end{equation}
		\end{lem}  
		\begin{proof}[Proof of Lemma \ref{1828}] 
			Fix $\theta_1,\theta_2$.  If   $R \geq  \frac n2 $,  \eqref{1220} is trivial. So   we consider  the case $R<\frac n 2$ in the following. 	Without loss of generality, we assume 	$0<\{q\omega\}<\frac12 $ and hence by \eqref{best},  \begin{equation}\label{jia}
				0<\{q\omega\}<\frac1q. 
			\end{equation}
			(The case $1-\frac1q <\{q\omega\}<1$ is treated by  the same argument.)  By Proposition \ref{DJC}, the orbit points $\{\{j\omega\}:\ 0\leq j\leq q-1\}$ leave gaps of length at most $\frac2q $.  Thus we can pick 
			$j^*\in\{0,\cdots,q-1\}$ such that 
			\begin{equation}\label{1954}
				0\leq\{\theta_2-(\theta_1+j^*\omega)\}<2/q.
			\end{equation}
			
			\medskip
			\noindent
			\textbf{Step 1: compare $\theta_1$ with $\theta_1+j^*\omega$.}
			Let $U:  (U \psi)(j)=\psi([j-1 ]_ n)$  be the unitary shift on $\ell^2([0;n-1])$.  Since  $R<\frac n 2$,    we can complete the $R$-truncated Toeplitz matrix $T_{\phi,n}^{(R)}$ to a circulant matrix $\widetilde{T}$ such that 
			\[
			\widetilde{T}=T_{\phi,n}^{(R)}+M,\quad U^{-1}\widetilde{T}U=\widetilde{T},
			\]
			where the boundary correction $M$ is supported near the two ends of $[0;n-1]$ and satisfies
			\[
			\operatorname{rank}M\le 2R.
			\]
			Hence, 
			\[
			U^{-j^*}T_{\phi,n}^{(R)}U^{j^*}-T_{\phi,n}^{(R)}=M-U^{-j^*}MU^{j^*},
			\]
			and therefore
			\begin{equation}\label{1909}
				\operatorname{rank}(U^{-j^*}T_{\phi,n}^{(R)}U^{j^*}-T_{\phi,n}^{(R)})\le 2\operatorname{rank}M\le 4R.
			\end{equation} 
			Note that 
			$$V_n(\theta_1+j^*\omega)-U^{-j^*}V_n(\theta_1) U^{j^*}=\operatorname{diag} (v(\theta_1+(j^*+j)\omega )-v(\theta_1+([j^*+j]_  n)\omega ))_{j\in [0;n-1]},$$
			where
			\begin{align}
				&\label{>}	v(\theta_1+(j^*+j)\omega )-v(\theta_1+([j^*+j]_  n)\omega ) \\ 
				=&\begin{cases}
					0, \quad &  0\leq j\leq n-j^*-1,\\
					v(\theta_1+(j^*+j)\omega )-v(\theta_1+(j^*+j)\omega-n\omega ), \quad & n-j^*\leq j\leq n-1. 
				\end{cases}\nonumber
			\end{align}
			Define  $$J_1=\{n-j^*\leq j\leq n-1:\  \{\theta_1+(j^*+j)\omega\}\in  [0,\{n\omega \}) \}. $$  For $j\notin J_1$,   we have $ \{\theta_1+(j^*+j)\omega-n\omega\}\leq \{\theta_1+(j^*+j)\omega\}$ and hence  $\eqref{>}\geq 0$  by the  monotonicity of $v$.  
			Since     $\{n\omega \}=\{sq\omega\}\leq s\{q\omega\} <s/q$ by \eqref{jia}  and  the $j^*<q$ points  
			$\theta_1+(j^*+j)\omega$ $(n-j^*\leq j\leq n-1)$ are mutually $1/(2q)$-separated by \eqref{best},  it follows that   $$\#J_1\leq \frac{s/q}{1/2q}+1=2s+1$$ and hence there exists a diagonal matrix $P_1$ with  $\operatorname{rank} P_1\leq\#J_1\leq 2s+1$ such that   \begin{equation}\label{ke55}
				U^{-j^*}V_n(\theta_1) U^{j^*}\leq V_n(\theta_1+j^*\omega)+P_1.
			\end{equation} 
			Combining \eqref{1909} and \eqref{ke55} yields there exists a self-adjoint matrix $P_2$  with $\operatorname{rank} P_2\leq 2s+1+4R$ such that 	\begin{equation}\label{ky}
				U^{-j^*}H_n^{(R)}(\theta_1) U^{j^*}\leq H_n^{(R)}(\theta_1+j^*\omega )+P_2, 
			\end{equation}

			\medskip
			\noindent
			\textbf{Step 2: compare $\theta_2$ with $\theta_1+j^*\omega$.}
			Define
			\[
			J_2=
			\{
			0\leq j\leq n-1:\ 
			\{\theta_1+j^*\omega+j\omega\}\in[1-2/q ,1)
			\}.
			\]
			For $j\notin J_2$, the short positive phase shift in \eqref{1954} does not cross a jump at the $j$th site. Thus $$
			v(\theta_2+j\omega)\geq v(\theta_1+j^*\omega+j\omega) ,\quad j\notin J_2.
			$$
			Writing $j=tq+a$, since for each fixed $0\le t<s$, the $q$ points
			$\theta_1+j^*\omega+tq\omega+a\omega$ $( 0\le a<q)$ 
			are mutually $1/(2q)$-separated by \eqref{best},
			$$
			\#J_2\le(\frac{2/q}{1/2q}+1)s=5s
			$$
			and   hence  
			there exists a diagonal matrix $P_3$ with  $\operatorname{rank} P_3\leq \#J_2\le5s$ such that 
			\begin{equation}\label{key4}
				H_n^{(R)}(\theta_1+j^*\omega)\leq H_n^{(R)}(\theta_2)+P_3.
			\end{equation}  Combining \eqref{ky} and \eqref{key4}  yields there exists a self-adjoint matrix $P_4$  with $\operatorname{rank} P_4\leq 7s+1+4R$ such that 
			$$	U^{-j^*}H_n^{(R)}(\theta_1) U^{j^*}\leq H_n^{(R)}(\theta_2)+P_4. $$   By the above inequality and Proposition \ref{pert},   
			\[
			N_{n}^{(R)}(\theta_2,E)
			\le
			N_{n}^{(R)}(\theta_1,E) +7s  +4R+1.
			\] 
			Interchanging $\theta_1,\theta_2$ finishes the proof of Lemma \ref{1828}. 
		\end{proof}
		The second lemma bounds $|N_n(\theta,E)-N_n^{(R)}(\theta,E)|$.  
		\begin{lem}\label{fankong}
			For $n=sq$,  we have 
			\begin{equation}\label{1020}
				\sup_{\theta,E}|N_n(\theta,E)-N_n^{(R)}(\theta,E)|\leq  sq (\frac{2\varepsilon e^{-\rho R}}{\rho })^{\frac 1\alpha}  +9(R+s+1)
			\end{equation}
			
		\end{lem}
		
		\begin{proof}[Proof of Lemma \ref{fankong}] 
			Fix $\theta$. Define  $$\eta_R:= (\frac{2\varepsilon e^{-\rho R}}{\rho })^{\frac 1\alpha} $$ and 
			$$J_3:=\{0\leq j\leq n-1:\ \{\theta+j\omega\} \in [1-\eta_R,1) \}.$$
			Since 	$$
			\|H_n(\theta)-H_{n}^{(R)}(\theta)\| \leq2 \varepsilon  \sum_{j>R}e^{-\rho j}\leq \frac{2\varepsilon e^{-\rho R}  }{\rho} =\eta_R^{\alpha},   
			$$ 
			by  the $\alpha$-H\"older monotonicity we have $$v(\theta+\eta_R+j\omega)-  v(\theta+j\omega)\geq \eta_R^{\alpha}\geq  \|H_n(\theta)-H_{n}^{(R)}(\theta)\|, \quad j\notin J_3$$
			By Proposition \ref{DJC}, 
			\[
			\#J_3
			\leq  n\eta_R +2s=s(q\eta_R +2) .
			\]
			Thus there exists a diagonal matrix $P_5$  with $\operatorname{rank} P_5\leq \#J_3 \leq s(q\eta_R +2)$ such that  $$ H_n(\theta) \leq  H_n^{(R)}(\theta)+ \eta_R^{\alpha} \operatorname{Id}_n  \leq  H_n^{(R)}(\theta+\eta_R)  +P_5.$$  It follows from the above inequality and Proposition \ref{DJC} that 
			\begin{equation}\label{keyf}
				N_n^{(R)}(\theta+\eta_R,E)\leq N_n(\theta,E)+ s(q\eta_R +2). 
			\end{equation} Combining \eqref{keyf} and \eqref{1220} yields  
			$$ N_n^{(R)}(\theta,E)\leq N_n(\theta,E)+ s(q\eta_R +2)+7(R+s+1).  $$
			Comparing  $H_n(\theta)$ with $H_n^{(R)}(\theta-\eta_R)$ by the  same argument gives  $$N_n(\theta,E)\leq N_n^{(R)}(\theta,E)+ s(q\eta_R +2)+7(R+s+1). $$ This  finishes the proof of Lemma \ref{fankong}. 
		\end{proof}
		Now we apply Lemmas \ref{1828} and \ref{fankong} to prove Lemma \ref{111}. 	Combining \eqref{1220} and \eqref{1020} yields for $n=sq$, 
		\begin{equation}\label{1258}
			\sup_{\theta_1,\theta_2,E}
			|N_n(\theta_1,E)-N_n(\theta_2,E)|
			\le
			2s q  (\frac{2\varepsilon e^{-\rho R}}{\rho })^{\frac 1\alpha}   +  30(R+s+1).
		\end{equation}
		Choosing $R=R(q)=2\lceil
		\frac\alpha \rho \ln q
		\rceil $ (so that $q e^{-(\rho R)/\alpha}<1$) in \eqref{1258} yields 
		$$	\sup_{\theta_1,\theta_2,E}
		|N_n(\theta_1,E)-N_n(\theta_2,E)|
		\le
		100( \frac\alpha \rho \ln q     + (1+(\frac{2\varepsilon}{\rho })^{\frac 1\alpha})s+1).$$ 
		This proves Lemma  \ref{111}.  \end{proof} 	
	\begin{proof}[Proof of Lemma \ref{l32}] 
		The condition $ l(t)<M$ is equivalent to $t\in[0,\vartheta_M)$.  The  count 	$\#\{0\le j\leq n-1: \{\theta +j\omega\}\in [0,\vartheta_M)  \}$ differs from $n\vartheta_M$ by at most 
		$2(s+|r|)$ by Proposition \ref{DJC}.  The number $M$ itself is the count  $\#\{0\le j\leq n-1:\vartheta_j \in [0,\vartheta_M)  \}=\#\{0\le j\leq n-1:\{-j\omega\} \in [0,\vartheta_M)  \}$  and therefore also differs
		from $n\vartheta_M$ by at most  $2(s+|r|)$ by Proposition \ref{DJC}. 
	\end{proof}
	\section{Lower bound for the denominator}\label{deno}
	In this section, we use the results from \S \ref{Rellich} to establish  a lower bound for $|P_n(\theta,E)|$, the  denominator of  Cramer's rule \eqref{cra}. 	
	First we define  the cutoff functions. 
	
	For 	$0<b<1<B$ (the choice of $b$ and $B$ will be specified  in \S \ref{ldt}), put  
	\begin{equation}\label{cut}
		\tau_B(t)
		=  \begin{cases}
			0,&|t|<B,\\
			\ln\frac{|t|}{B},&|t|\ge B, 
		\end{cases}\quad 
		g_B(t)= \begin{cases}
			\ln |t|,&|t|< B,\\
			\ln B,&|t|\ge B, 
		\end{cases}
		\quad 
		h_{b,B}(t)=
		\begin{cases}
			\ln b,&|t|< b,\\
			\ln|t|,&b\leq |t|<B,\\
			\ln B,&|t|\ge B.
		\end{cases}
	\end{equation}
	
	\begin{prop}
		\label{232}
		There exists some   $C_2=C_2(\alpha,\rho ,\varepsilon)>0$ such that for $n=sq+r$,    $\theta\in \T$ and $E\in \R$, 
		\begin{align}
			\ln |P_n(\theta,E)|
			\geq&
			\ln  |P_n^{<b}(\theta,E) |
			\nonumber\\
			&+
			nJ_B(E)
			-
			n\kappa_\alpha(A_\rho\varepsilon)
			-
			C_2(\ln q+s+|r|)
			\ln\frac{B}{b}
			\label{m2} \\
			&+
			T_{n,B}(\theta,E)
			-
			\frac{nA_\rho\varepsilon}{B},\label{m3}
		\end{align}
		where 	$P_n^{<b}(\theta,E)$ was defined in \eqref{ps},  $A_\rho$ was defined in \eqref{633}, 
		\begin{equation}\label{pm}
			J_B(E)
			:=
			\int_0^1 g_B(v(\theta)-E)d\theta,
		\end{equation}
		\begin{equation}\label{pl}
			T_{n,B}(\theta,E)
			:=
			\sum_{k=0}^{n-1}
			\tau_B(v(\theta+k\omega)-E)
		\end{equation}
		and the function \(\kappa_\alpha:[0,+\infty)\to[0,+\infty)\)  is defined as 	\begin{equation}\label{kap}
			\kappa_\alpha(a)
			:=	3\alpha \,a^{\frac1{2\alpha}}, 
		\end{equation}
		
	\end{prop}
	
	\begin{rem}\label{1023}
		Proposition \ref{232} gives the lower bound for
		$\ln|P_n(\theta,E)|$ according to the three types of finite-volume
		eigenvalues aforementioned.
		
		The first term
		$
		\ln|P_n^{<b}(\theta,E)|
		$
		is the contribution of the small eigenvalues, namely those satisfying
		$|E_k(\theta)-E|<b$. These factors can be abnormally small and cannot
		be effectively canceled by the corresponding contribution from the
		numerator. Therefore, they have to be kept explicitly in the lower
		bound and treated separately.
		
		The terms in \eqref{m2} describe the contribution of the
		middle eigenvalues. Their main contribution is
		$nJ_B(E)$. The remaining terms in \eqref{m2} are the errors arising when
		the finite-volume eigenvalues are replaced by the values of the
		principal branch $\Lambda_0$  along the quasiperiodic orbit and then compared with
		the potential $v$.
		
		Finally, the terms in \eqref{m3} describe the contribution of the
		large eigenvalues. Their main contribution is 
		$T_{n,B}(\theta,E)$, up to the error $nA_\rho\varepsilon/B$. 
	\end{rem}
	
	\begin{proof}[Proof of Proposition \ref{232}] 
		Recall that 
		\[
		P_n(\theta,E)
		=
		\det(H_n(\theta)-E)
		=
		\prod_{k=0}^{n-1}
		(E_k(\theta)-E),
		\]
		and hence
		\[
		\ln |P_n(\theta,E)|
		=
		\sum_{k=0}^{n-1}
		\ln |E_k(\theta)-E |.
		\]
		For every $t\in \R$, we have 
		\[
		\ln|t|
		=
		\mathbf 1_{\{|t|<b\}}
		(\ln|t|-\ln b) 
		+ 	h_{b,B}(t)
		+
		\tau_B(t).
		\]	Applying this identity with
		$t=E_k(\theta)-E$ and summing over
		$k=0,\dots,n-1$, we obtain
		\[
		\ln |P_n(\theta,E)|
		=
		\Sigma_1+\Sigma_2+\Sigma_3,
		\]
		where
		\begin{align*}
			\Sigma_1
			&:=
			\sum_{\substack{0\leq k\leq n-1\\
					|E_k(\theta)-E|<b}}
			(
			\ln|E_k(\theta)-E|
			-\ln b
			),\\
			\Sigma_2
			&:=
			\sum_{k=0}^{n-1}
			h_{b,B}
			(E_k(\theta)-E),\\
			\Sigma_3
			&:=
			\sum_{k=0}^{n-1}
			\tau_B
			(E_k(\theta)-E).
		\end{align*}
		
		We estimate these three terms separately.
		
		\medskip
		\noindent
		\textbf{Estimate of $\Sigma_1$.}
		Since $0<b<1$, one has $\ln b<0$, and  hence 
		\begin{equation}
			\Sigma_1
			\geq
			\ln |P_n^{<b}(\theta,E)|.
			\label{sum1}
		\end{equation}
		
		\medskip
		\noindent
		\textbf{Estimate of $\Sigma_2$.}
		Applying  Corollary \ref{BV}  with
		$g=h_{b,B}$ and noting $\operatorname{Var}(h_{b,B})\leq 2\ln(B/b)$, one has 
		\begin{align}
			\Sigma_2
			\geq
			\sum_{k=0}^{n-1}
			h_{b,B}
			(
			\Lambda_0(\theta+k\omega)-E
			)
			-
			C_1(\ln q+s+|r|)
			\ln\frac{B}{b}.
			\label{333}
		\end{align}
		Since $\Lambda_0(\theta)$ is monotone on $[0,1)$ (cf. Proposition \ref{mon}),	we have 
		$$\operatorname{Var}_\T	(h_{b,B}(\Lambda_0(\cdot)-E))\leq \operatorname{Var}(h_{b,B})\leq 2\ln(B/b).$$  Hence, by  Proposition \ref{dj}, 
		\begin{align}\label{222}
			\sum_{k=0}^{n-1}
			h_{b,B}
			(\Lambda_0(\theta+k\omega)-E)
			\geq
			n\int_0^1
			h_{b,B}
			(\Lambda_0(\theta)-E)d\theta
			-
			2(s+|r|)
			\ln\frac{B}{b}.
		\end{align}
		Since
		$	h_{b,B}(t)\geq g_B(t)$
		for all $t$, it follows from \eqref{333} and \eqref{222} that 
		\begin{equation}\label{56}
			\Sigma_2
			\geq
			n	\int_0^1
			g_B
			(\Lambda_0(\theta)-E)d\theta-	C_2(\ln q+s+|r|)
			\ln\frac{B}{b} .
		\end{equation}
		We now compare the integral in \eqref{56} with
		\[
		J_B(E)
		=
		\int_0^1g_B(v(\theta)-E)d\theta.
		\]
		Recall that
		$	|\Lambda_0(\theta)-v(\theta)|
		\le A_\rho\varepsilon
		$ (cf. \eqref{duibi}).
		Using 
		\begin{equation}\label{zy}
			g_B(|t|+a)-g_B(t)
			\le
			\ln\frac{|t|+a}{|t|}
			=
			\ln(1+\frac{a}{|t|}), \quad a>0,
		\end{equation}
		we obtain  
		$$g_B(v(\theta)-E) \leq 	g_B
		(|\Lambda_0(\theta)-E|+A_\rho\varepsilon)\leq g_B
		(\Lambda_0(\theta)-E)+\ln(
		1+
		\frac{A_\rho\varepsilon}{|\Lambda_0(\theta)-E|}
		) $$
		and hence 
		\begin{equation}\label{k35}
			J_B(E)	-	\int_0^1g_B(\Lambda_0(\theta)-E)d\theta	\leq	\int_0^1	\ln(	1+	\frac{A_\rho\varepsilon}{|\Lambda_0(\theta)-E|}	)d\theta.
		\end{equation}
		By $ |\Lambda_0(\theta)-E|\geq \|\theta-\widetilde z_{0}(E)\|_\T^\alpha$ (cf.  \eqref{chengj}), 
		$$
		\int_0^1
		\ln(
		1+
		\frac{A_\rho\varepsilon}{|\Lambda_0(\theta)-E|}
		)d\theta
		\le
		\int_0^1
		\ln(
		1+\frac{A_\rho\varepsilon}{\|\theta-\widetilde z_{0}(E)\|_\T^\alpha}
		)d\theta
		\le
		\int_{-1}^{1}
		\ln(
		1+\frac{A_\rho\varepsilon}{|t|^\alpha}
		)dt.
		$$
		We claim that
		\begin{equation}\label{yy1}
			\int_{-1}^{1}
			\ln(
			1+\frac{a}{|t|^\alpha}
			)dt
			\leq (4\ln2)  \alpha  \,  a^{\frac{1}{2\alpha}}  \leq 
			3 \alpha\,a^{\frac1{2\alpha}}=:\kappa_\alpha(a),  \quad a>0. 
		\end{equation}
		Indeed, for every $x>0$ and $\alpha\geq 1$,
		$	\ln(1+x^\alpha)
		\leq
		\alpha \ln(1+x)  .$ 
		Taking
		$	x=a^{\frac{1}{\alpha}}/|t|$,  we get
		$$
		\int_{-1}^{1}
		\ln(
		1+\frac{a}{|t|^\alpha}
		)dt
		\leq 	2\alpha  	\int_{0}^{1}\ln(
		1+\frac{a^{\frac{1}{\alpha}}}{|t|}
		)dt
		=	2\alpha  ((1+ a^{\frac{1}{\alpha}} )\ln (1+a^{\frac{1}{\alpha}})-a^{\frac{1}{\alpha}}\ln a^{\frac{1}{\alpha}})\leq (4\ln2)  \alpha \,  a^{\frac{1}{2\alpha}} ,
		$$ where we used the inequality $$(1+x)\ln(1+x) -x\ln x\leq (2\ln2)\sqrt{x}, \quad x>0.$$   
		From \eqref{k35} and the claim, we have proved
		\begin{equation}
			\int_0^1
			g_B(\Lambda_0(\theta)-E)d\theta
			\ge
			J_B(E)-\kappa_\alpha(A_\rho\varepsilon).
			\label{fafa}
		\end{equation}
		Substituting \eqref{fafa} into
		\eqref{56}, we conclude that
		\begin{equation}
			\Sigma_2
			\ge
			nJ_B(E)
			-
			n\kappa_\alpha(A_\rho\varepsilon)
			-
			C_2(\ln q+s+|r|)
			\ln\frac{B}{b}.
			\label{sum2}
		\end{equation}
		\medskip
		\noindent
		\textbf{Estimate of $\Sigma_3$.}
		Recall $\|H_n(\theta)-V_n(\theta)\|
		\leq A_\rho\varepsilon$ (cf. \eqref{duibi}).
		Let
		$
		d_0(\theta)\leq\cdots\leq d_{n-1}(\theta)
		$
		be the non-decreasing ordered diagonal values
		$\{v(\theta+k\omega)\}_{k=0}^{n-1}$.
		By Proposition \ref{sta},   
		$$	 |E_k(\theta)-d_k(\theta)	 | \leq \|H_n(\theta)-V_n(\theta)\|	\leq A_\rho\varepsilon,	\quad	0\leq k\leq n-1.
		$$
		Since  function
		$	\tau_B 
		$ is $B^{-1}$-Lipschitz, we have 
		\[
		\tau_B(E_k(\theta)-E)
		\geq
		\tau_B(d_k(\theta)-E)
		-
		\frac{A_\rho\varepsilon}{B}.
		\]
		Summing in $k$ yields
		\begin{equation}\label{sum3}
			\Sigma_3
			\geq
			\sum_{k=0}^{n-1}
			\tau_B(d_k(\theta)-E)
			-
			\frac{nA_\rho\varepsilon}{B}=  T_{n,B}(\theta,E)	-
			\frac{nA_\rho\varepsilon}{B} .
		\end{equation}
		Finally, combining
		\eqref{sum1},		\eqref{sum2}		and	\eqref{sum3}, we finish the proof of Proposition \ref{232}. 
	\end{proof}

	\section{Upper bound for the numerator}\label{num}
	In this section, we establish an upper bound for the numerator in
	Cramer's rule \eqref{cra}.
	We first recall  the following lemma from \cite{BJ02},  which  reduces the study of cofactor to the study of determinants of restrictions of
	$H(\theta)-E$.     
	\begin{lem}[{\cite[Lemma 10]{BJ02}}]  
		\label{cof}
		For $x,y\in I_n$, 
		\begin{equation}
			\operatorname{cof}_{x,y}(H_{n}(\theta)-E)
			=
			\sum_{\Gamma:x\to y}a_\Gamma \varepsilon^{\ell(\Gamma)-1} 
			\det(H_{I_n\setminus\Gamma}(\theta)-E) 
			\prod_{l=1}^{\ell(\Gamma)-1}
			\phi_{\gamma_{l+1}-\gamma_l},
		\end{equation}
		where the sum is taken over all simple paths  (ordered subsets)   $\Gamma=(\gamma_1,\cdots,\gamma_{\ell(\Gamma)})$ of $I_n$  with  $x=\gamma_1$ to $y=\gamma_{\ell(\Gamma)}$ and  $a_\Gamma\in \{\pm1\}$ depends on $\Gamma$. 
	\end{lem} 
	To apply the above lemma, we need an upper bound on the determinants
	of restrictions of $H(\theta)-E$ to arbitrary subsets. The following
	proposition provides such an estimate.
	\begin{prop}  	\label{Had}
		Let
		$
		n=sq+r,$ and let $a_\varepsilon$ and $B$ satisfy  	\begin{equation}\label{ae}
			A_\rho\varepsilon\leq a_\varepsilon\leq B.
		\end{equation}
		Then  for every $\Gamma\subset I_n$,  we have 
		\begin{align}
			\nonumber 	\ln|	\det(H_{I_n\setminus\Gamma}(\theta)-E) |
			\leq&  -\ell(\Gamma)  \ln a_\varepsilon\\	
			&	+	nJ_B(E)+	n\kappa_{\alpha}(a_\varepsilon)
			+  	
			2(s+|r|)\ln\frac{B}{a_\varepsilon}\label{zi2}
			\\
			&+ 	T_{n,B}(\theta,E)
			+ \frac{na_\varepsilon}{B}, \label{zi3}
		\end{align} 
		where  $J_B(E)$, $T_{n,B}(\theta,E)$ and $\kappa_{\alpha}(\cdot)$ were defined in \eqref{pm}, \eqref{pl} and \eqref{kap}. 
	\end{prop}
	
	\begin{proof}
		Since $\varepsilon\|\phi\|_{\ell^1}\leq A_\rho\varepsilon\leq a_\varepsilon$ (cf. \eqref{566}),  	Hadamard's inequality gives
		\[
		|\det(H_{I_n\setminus\Gamma}(\theta)-E)|
		\le
		\prod_{k\in I_n\setminus\Gamma}
		\|\operatorname{row}_k(H_{I_n\setminus\Gamma}(\theta)-E)\|_{\ell^1}\leq \prod_{k\in I_n\setminus\Gamma}
		(|v(\theta+k\omega)-E|+a_\varepsilon) .
		\]
		Thus 
		\begin{align}
			\ln|	\det(H_{I_n\setminus\Gamma}(\theta)-E) |
			\leq  &	\sum_{k\in I_n\setminus\Gamma}
			\ln(|v(\theta+k\omega)-E|+a_\varepsilon) \nonumber\\ \leq   &\sum_{k=0}^{n-1}
			\ln(|v(\theta+k\omega)-E|+a_\varepsilon) - 	\ell(\Gamma)\ln a_\varepsilon  \label{1103}
		\end{align}
		Recall the cut off functions (cf. \eqref{cut}). 	Note that $$\ln (|t|+a_\varepsilon)\leq g_B(|t|+a_\varepsilon)+\tau_B(t)+\frac {a_\varepsilon} B.$$
		Using  the above inequality   to estimate  the first summation in \eqref{1103} yields 
		\begin{align}
			\nonumber	\sum_{k=0}^{n-1}
			\ln(|v(\theta+k\omega)-E|+a_\varepsilon) & \leq   	 \sum_{k=0}^{n-1}g_{B}(|v(\theta+k\omega)-E|+a_\varepsilon)+  \sum_{k=0}^{n-1}\tau_{B}(v(\theta+k\omega)-E)+ \frac{na_\varepsilon}{B}\\
			&=  \sum_{k=0}^{n-1}g_{B}(|v(\theta+k\omega)-E|+a_\varepsilon)  + 	T_{n,B}(\theta,E)
			+ \frac{na_\varepsilon}{B}. \label{46}
		\end{align}
		Since $v$ is monotone, we have  $\operatorname{Var}_\T( g_{B}(|v(\cdot)-E|+a_\varepsilon) \leq \operatorname{Var}(g_B(|\cdot|+a_\varepsilon))\leq 2\ln(B/a_\varepsilon)$. It follows from Proposition \ref{dj}   that 
		\begin{equation}\label{45}
			\sum_{k=0}^{n-1}g_{B}(|v(\theta+k\omega)-E|+a_\varepsilon)
			\leq 
			n\int_0^1g_{B}(|v(\theta)-E|+a_\varepsilon)d\theta 
			+2(s+|r|)\ln\frac B {a_\varepsilon}.
		\end{equation}
		By \eqref{zy}, \eqref{yy1} and the $\alpha$-H\"older monotonicity of $v$,  
		\begin{align*}
			\int_0^1g_{B}(|v(\theta)-E|+a_\varepsilon)d\theta 
			&\leq   	\int_0^1g_{B}(|v(\theta)-E|)d\theta  +  \ln(	1+	\frac{a_\varepsilon}{|v(\theta)-E|}	)d\theta\\
			&	\leq  J_B(E)+	\int_{-1}^{1}
			\ln(
			1+\frac{a_\varepsilon}{|t|^\alpha}
			)dt\\
			&\leq J_B(E)+\kappa_{\alpha}(a_\varepsilon). 
		\end{align*}
		Combining the above inequality with \eqref{1103},  \eqref{46} and \eqref{45}, we finish the proof of Proposition  \ref{Had}. 
	\end{proof}
	Now we combine  Lemma \ref{cof} and Proposition \ref{Had} to prove the following upper bound for the numerator. 
	\begin{prop}\label{Had2}
		Under the same condition of Proposition \ref{Had},  assume moreover\begin{equation}\label{xue}
			\rho-\ln(1+\frac{2\varepsilon}{a_\varepsilon})>0.
		\end{equation}  Then  there exists some $C_3=C_3(a_\varepsilon, \rho-\ln(1+\frac{2\varepsilon}{a_\varepsilon}) )>0$ such that 
		$$	|\operatorname{cof}_{x,y}(H_{n}(\theta)-E)| \leq  C_3 e^{G_n(\theta,E)}e^{-(\rho-\ln(1+\frac{2\varepsilon}{a_\varepsilon}))|x-y|}, $$
		where $$
		G_n(\theta,E):=\eqref{zi2}+\eqref{zi3}=nJ_B(E)+n\kappa_\alpha(a_\varepsilon)
		+2(s+|r|)\ln\frac{B}{a_\varepsilon}
		+T_{n,B}(\theta,E)+\frac{na_\varepsilon}{B}.
		$$
	\end{prop}
	\begin{proof}
		By \eqref{lr},  Lemma \ref{cof}  and Proposition \ref{Had},  we have
		\begin{align}
			\left|
			\operatorname{cof}_{x,y}(H_{n}(\theta)-E)
			\right|
			&\le
			\sum_{\Gamma:x\to y}
			\varepsilon^{\ell(\Gamma)-1}
			\left|\det(H_{I_n\setminus\Gamma}(\theta)-E)\right|
			\prod_{j=1}^{\ell(\Gamma)-1}
			|\phi_{\gamma_{j+1}-\gamma_j}|\nonumber \\
			&\le
			a_\varepsilon^{-1}e^{G_n(\theta,E)}
			\sum_{\Gamma:x\to y}
			(\frac{\varepsilon}{a_\varepsilon})^{\ell(\Gamma)-1}
			e^{-\rho L(\Gamma)},\label{dai}
		\end{align}
		where
		\[
		L(\Gamma):=
		\sum_{j=1}^{\ell(\Gamma)-1}
		|\gamma_{j+1}-\gamma_j|
		\ge |x-y|.
		\]
		For fixed \(L\ge0\) and \(1\le \ell\le L+1\), the number of paths
		satisfying \(L(\Gamma)=L\) and \(\ell(\Gamma)=\ell\) is bounded by
		$
		2^{\ell-1}\binom{L}{\ell-1}
		$
		since the \(\ell-1\) positive integers
		\(
		|\gamma_{j+1}-\gamma_j|
		\)
		form a composition of $L$ and  each increment has at most two
		possible signs.   Hence
		$$	\sum_{\substack{\Gamma:x\to y\\ L(\Gamma)=L}}
		(\frac{\varepsilon}{a_\varepsilon})^{\ell(\Gamma)-1}
		\leq
		\sum_{\ell=1}^{L+1}
		2^{\ell-1}\binom{L}{\ell-1}	(\frac{\varepsilon}{a_\varepsilon})^{\ell -1} 
		\leq   (1+	\frac{2\varepsilon}{a_\varepsilon})^L.$$ 
		By the above equality and  \eqref{xue}, 
		\begin{equation}\label{ru}
			\sum_{\Gamma:x\to y}
			(\frac{\varepsilon}{a_\varepsilon})^{\ell(\Gamma)-1}e^{-\rho L(\Gamma)}
			\le
			\sum_{L\ge |x-y|}
			e^{-\rho L}(1+2\frac{\varepsilon}{a_\varepsilon})^{L} 
			\leq \frac{2}{1-e^{-(\rho-\ln(1+\frac{2\varepsilon}{a_\varepsilon}))}}	e^{-(\rho-\ln(1+\frac{2\varepsilon}{a_\varepsilon}))|x-y|}. 
		\end{equation}
		Substituting \eqref{ru} to \eqref{dai}, we finish the proof of Proposition \ref{Had2}. 
	\end{proof}

	\section{Large deviation theorem}\label{ldt}
	In this section, we combine the denominator and numerator estimates
	obtained in the previous sections to prove a large deviation theorem for the Green's function. We also explain the choice of $\varepsilon_0(\delta,\alpha)$ in Theorem \ref{main}.
	
	Putting \eqref{cra}, Propositions~\ref{232} and~\ref{Had2} together
	(assume temporarily that \eqref{ae} and \eqref{xue} hold) and 	noticing  that the common terms $nJ_B(E)$ and	$T_{n,B}(\theta,E)$ cancel in the ratio, we obtain
	\begin{equation}\label{mean}
		|G_{I_n}^{\theta,E}(x,y)|
		\leq
		C_3
		\frac{e^{nR_n}}{|P_n^{<b}(\theta,E)|}
		e^{-(\rho-\ln(1+\frac{2\varepsilon}{a_\varepsilon}))|x-y|},
	\end{equation}
	where
	\begin{align}
		R_n
		&:=
		\frac1n(\eqref{zi2}+\eqref{zi3}-\eqref{m2}-\eqref{m3})
		=K_\varepsilon+\mathcal E_{n,q},
		\label{1111}\\
		K_\varepsilon
		&:=
		\kappa_\alpha(A_\rho\varepsilon)
		+\kappa_\alpha(a_\varepsilon),\nonumber
		\label{pra2}\\
		\mathcal E_{n,q}
		&:=
		\frac{a_\varepsilon+A_\rho\varepsilon}{B}
		+\frac{C_2(\ln q+s+|r|)}{n}\ln\frac{B}{b}
		+\frac{2(s+|r|)}{n}\ln\frac{B}{a_\varepsilon}.
		\nonumber
	\end{align}
	In the right-hand side of  \eqref{mean}, several terms may lead to a loss in the decay rate $\rho$ coming from the long-range hopping.  	We call the terms $K_\varepsilon$ and $\ln(1+\frac{2\varepsilon}{a_\varepsilon})$  the \textbf{perturbative loss}; they  can be made small by taking 	$\varepsilon$ sufficiently small.  We call the  term $\mathcal E_{n,q}$ the \textbf{finite-scale loss};   after a suitable choice 	of the auxiliary parameters $b$ and $B$, it tends to zero as $n$ tends to infinity. 
	
	This mechanism is similar in spirit to the Schr\"odinger case in
	\cite{JK24}. There,  the contributions of the middle eigenvalues of both numerator and denominator are estimated through the Lyapunov exponent, without an additional perturbative loss, while the finite-scale loss tends to zero. The remaining difficulty comes from the smallness of   $|P_n^{<b}(\theta,E)|$, which may produce loss of size $\beta$ due to the continued-fraction 
	properties of $\omega$ in the subsequent interpolation argument. We call it  the
	\textbf{arithmetic loss}. Thus, in the Schr\"odinger case, localization depends on 
	a  competition between the Lyapunov exponent $L(E)$ (natural exponential  decay rate for Green's function)  and the
	arithmetic loss $\beta$.
	
	In the present long-range setting, the finite-scale loss tends to zero as in \cite{JK24} but  we do not have such an exact cancellation for middle eigenvalues  based on the Lyapunov exponent, leaving the
	additional perturbative  losses.   Moreover, the arithmetic loss will be of size $\alpha\beta$ due to the weaker $\alpha$-H\"older monotonicity condition. 
	
To obtain meaningful Green's function estimates,   we need 	to retain a positive part of   $\rho$ after these losses	are taken into account. A sufficient condition is 
	\[
	\rho-\ln(1+\frac{2\varepsilon}{a_\varepsilon})-\alpha\beta-4K_\varepsilon>0.
	\]
	Recall that
	\[
	\delta=\rho-\alpha\beta>0.
	\]
	Thus, it suffices to choose $\varepsilon$  small so that the losses
	$K_\varepsilon$  and   $\ln(1+\frac{2\varepsilon}{a_\varepsilon})$   use only a  portion
	of $\delta$. More precisely, we require
	\begin{equation}\label{1n}
		\ln(1+\frac{2\varepsilon}{a_\varepsilon})
		\leq\frac14\delta,
		\quad
		K_\varepsilon\leq\frac1{16}\delta.
	\end{equation}
	Choose
	\begin{align}
		\varepsilon\leq\varepsilon_0
		&:=
		\frac12
		(\frac{\delta}{100\alpha})^{2\alpha}
		\cdot (e^{\frac\delta4}-1),
		\label{es0}\\
		a_\varepsilon
		&:=
		(\frac{\delta}{100\alpha})^{2\alpha}.
		\label{ave}
	\end{align}
	Then \eqref{1n}, as well as the assumptions \eqref{ae} and \eqref{xue},
	holds for sufficiently large $B$. Indeed,
	$A_\rho\varepsilon\leq a_\varepsilon$ because $\delta<\rho$ and	$\frac12 A_\rho (e^{\delta/4}-1)<1.$
	Moreover,
	$
	\frac{2\varepsilon}{a_\varepsilon}
	\leq e^{\delta/4}-1,
	$
	and
	$
	K_\varepsilon= \kappa_\alpha(A_\rho\varepsilon)
	+\kappa_\alpha(a_\varepsilon)
	\leq2\kappa_\alpha(a_\varepsilon)\leq  \frac{6\delta}{100}
	<\frac{\delta}{16}.
	$
	
	Choosing 
	\begin{equation}\label{canshu}
		b=b(q):=q^{-10\alpha},
		\quad
		B=B(q):=q
	\end{equation}
and $r=-1$, 	we obtain the following  large deviation theorem for $n=sq-1$. In fact, one  can obtain a large deviation theorem for any scale $n=sq+r$. However, the choice $r=-1$ yields the estimates relevant to our purposes and is sufficient for the proof of Anderson localization.
	\begin{thm}[Large deviation theorem for $n=sq-1$]\label{LDT}
		Let $0\leq\varepsilon\leq\varepsilon_0$, where $\varepsilon_0$ is
		defined in \eqref{es0}.  Then there exist constants $C_3,C_4>0$ (depending on $\rho,\delta,\alpha$), such that for $n=sq-1$, 
		\begin{equation}\label{1222}
			|G_{I_n}^{\theta,E}(x,y)|
			\leq
			C_3
			\frac{
				e^{(\frac1{16}\delta+
					C_4\frac{(\ln q)^2}{q})n}
			}{
				|P_n^{<b(q)}(\theta,E)|
			}
			e^{-(\rho-\frac14\delta)|x-y|}.
		\end{equation}
		Consequently, for sufficiently large $q$
		(depending on $\rho,\delta,\alpha$) and $n=sq-1$, 
		\begin{equation}\label{yaoyong}
			|G_{I_n}^{\theta,E}(x,y)|
			\leq
			\frac{
				e^{\frac1{15}\delta n}
			}{
				|P_n^{<b(q)}(\theta,E)|
			}
			e^{-(\rho-\frac14\delta)|x-y|}.
		\end{equation}
	\end{thm}
	
	\begin{proof}
		Since $0\leq\varepsilon\leq\varepsilon_0$, the choice of
		$a_\varepsilon$ in \eqref{ave} gives \eqref{1n}. Moreover, by the
		choice of $b(q)$, $B(q)$ (cf. \eqref{canshu}) and $|r|=1$, we have
		\begin{equation}\label{1950}
			\mathcal E_{n,q}
			\leq
			C_4\frac{(\ln q)^2}{q}.
		\end{equation}
		Substituting \eqref{1111}, \eqref{1n} and \eqref{1950} into
		\eqref{mean}, we obtain \eqref{1222}. The estimate \eqref{yaoyong}
		follows by taking $q$ sufficiently large so that 
		$C_3	e^{(\frac1{16}\delta+C_4\frac{(\ln q)^2}{q})n}\leq 	e^{\frac1{15}\delta n}.	$
	\end{proof}
	
	\begin{rem}
		The terminology ``large deviation theorem'' can be understood as
		follows. Away from exceptional phases, one expects the Green's function
		to have off-diagonal exponential decay rate  $\rho$ with only a small loss. In the present setting, the expected estimate is of the form
		\[
		|G_{I_n}^{\theta,E}(x,y)|
		\leq
		e^{-(\rho-\frac14\delta)|x-y|+\frac1{15}\delta n}.
		\]
		Theorem \ref{LDT} shows that a significant deviation from this expected
		estimate can only come from the smallness of  quantity
		$|P_n^{<b(q)}(\theta,E)|$. 
		By the quasi-zero analysis in \S \ref{Rellichqz} (cf. \eqref{sf}), such a  deviation  can occur only when $\theta$ is sufficiently close to the
		quasi-zeros, which is
		exactly the form needed for using the interpolation  argument  in
		the next section.
	\end{rem}

	\section{Anderson localization}\label{al}
The proof of Anderson localization is based on the  Green's function method. 
The key argument is to   combine  the large deviation theorem  with an interpolation argument  to construct good intervals on which the Green's function satisfies
	a suitable exponential estimate. When  $\beta>0$, the construction of these good
	intervals requires a distinction between resonant and nonresonant regions of $\Z$, based on a careful analysis of the continued-fraction
	approximations of $\omega$. Once such good intervals are constructed,
	the rest of the proof follows from the  Schnol's argument and Poisson formula.  Such arguments  originated in \cite{Jit94,Jit99} and developed in \cite{JL18,JL24,JK24}.   In this section, we extend the arguments to the long-range setting to prove Theorem \ref{main}.

	\subsection{Regular points and  good intervals}  In this subsection, we give  the definition  of regular point and good   interval.	\begin{defn}\label{rd}
		Let $n=sq-1$, $\theta\in\T$, $E\in\R$, and $\mu\in\R$.
		Recall that $b(q)=q^{-10\alpha}$ (cf. \eqref{canshu}).
		We call a point $x\in\Z$
		$(q,n,\theta,E,\mu)$-\textbf{good} if
		\[
		|P_n^{<b(q)}(\theta+x\omega,E)|\geq e^{-\mu n}.
		\]
		We call a point $z\in\Z$
		$(q,n,\theta,E,\mu)$-\textbf{regular} if there exists some
		\begin{equation}\label{ac}
			x\in
			[z-\lfloor\frac34n\rfloor;
			z-\lfloor\frac14n\rfloor]
		\end{equation}
		such that $x$ is $(q,n,\theta,E,\mu)$-good. 	We will call the $\mu $  the regularity  parameter of $x$. 
	\end{defn}
	By Theorem~\ref{LDT}, we can associate a good  interval with each
	regular point.
	
	\begin{prop}\label{2116}
		Let $n=sq-1$, $\theta\in\T$, $E\in\R$, and $\mu\in\R$.
		For  sufficiently large $q$ 	(depending on $\rho,\delta,\alpha$), the following
		holds. If $z\in\Z$ is $(q,n,\theta,E,\mu)$-regular, then there exists
		an integer interval $I(z)$ of length $n$ such that
		\begin{equation}\label{juli}
			\operatorname{dist}(z,\Z\setminus I(z))
			\geq \frac14 n
		\end{equation}
		and
		\begin{equation}\label{yaozheng}
			|G_{I(z)}^{\theta,E}(x,y)|
			\leq
			e^{(\frac{1}{15}\delta+\mu)n}
			e^{-(\rho-\frac{1}{4}\delta)|x-y|},
			\quad x,y\in I(z).
		\end{equation}
	\end{prop}
	
	\begin{proof}
		By the definition of $(q,n,\theta,E,\mu)$-regularity, there exists a
		$(q,n,\theta,E,\mu)$-good point
		\begin{equation}\label{tuijuli}
			z'\in
			[z-\lfloor\frac{3}{4}n\rfloor;
			z-\lfloor\frac{1}{4}n\rfloor]
		\end{equation}
		such that
		\begin{equation}\label{jiesuan}
			|P_n^{<b(q)}(\theta+z'\omega,E)|
			\geq e^{-\mu n}.
		\end{equation}
		Set
		\[
		I(z):=[z';z'+n-1].
		\]
		Then $I(z)$ is an integer interval of length $n$, and
		\eqref{tuijuli} implies \eqref{juli}. Moreover, for sufficiently large $q$,  combining
		\eqref{yaoyong} with \eqref{jiesuan} yields 
		$$
		|G_{I_n}^{\theta+z'\omega,E}(x,y)|
		\leq
		e^{(\frac{1}{15}\delta+\mu)n}
		e^{-(\rho-\frac{1}{4}\delta)|x-y|}.
		$$
		By \eqref{mosa}, we can identify 
		$H_{n}(\theta+z'\omega)$ with $H_{I(z)}(\theta)$.
		Hence, the above estimate gives \eqref{yaozheng}.
	\end{proof}
	
	\begin{defn}\label{gq}
		We call an integer interval $I$ of length $n$
		$(n,\theta,E,\mu)$-\textbf{good} if
		\begin{equation}\label{ig}
			|G_I^{\theta,E}(x,y)|
			\leq
			e^{(\frac{1}{15}\delta+\mu)n}
			e^{-(\rho-\frac{1}{4}\delta)|x-y|},
			\quad x,y\in I.
	\end{equation} \end{defn}
	With this terminology, Proposition \ref{2116} can be restated as
	follows. For  sufficiently large $q$ and
	$n=sq-1$, every $(q,n,\theta,E,\mu)$-regular point $z$ is associated
	with an $(n,\theta,E,\mu)$-good  interval $I(z)$ such
	that
	$	\operatorname{dist}(z,\Z\setminus I(z))\geq\frac14 n.$
	\begin{rem}
		Our definition of regularity is	slightly different from that in \cite{JK24} and the related	literature. There, the regularity parameter is usually used to	describe directly the off-diagonal decay of the Green's function. Here, instead, $\mu$ is defined through the lower	bound for  $P_n^{<b(q)}$. Note that a larger $\mu$ corresponds	to a weaker lower bound on $|P_n^{<b(q)}|$, and hence a worse Green's function estimate.
	\end{rem}
	
	\subsection{Statement of the regularity result and proof of  Theorem \ref{main}}
	In this subsection, we first  state the regularity result,
	Proposition \ref{regu}, which establishes different  types of 
	regularity for resonant and nonresonant points.   We then  use this
	result to construct good intervals for the Poisson formula iteration and then prove  Theorem \ref{main}.
	The proof of Proposition \ref{regu} will be postponed to the next
	section, since it involves some arguments based on the Poisson formula
	iteration that are similar to those used in this subsection. 
	
	We first recall Schnol's theorem; see, for
	example, \cite{Han19}. It implies that, to prove Anderson localization for
	$H(\theta)$, it suffices to prove exponential decay for  every
	polynomially bounded generalized eigenfunction. More
	precisely, it suffices to show that for some $\rho'>0$, 
	\begin{equation}\label{aim}
		H(\theta)\psi=E\psi,\quad
		|\psi(z)|\leq 1+|z|
		\quad\Longrightarrow\quad
		|\psi(z)|\leq C(\psi)  e^{-\rho'|z|}
	\end{equation}
	
	In the rest of this section, we fix a generalized eigenvalue $E^*$ of $H(\theta^*)$ and a corresponding nonvanishing polynomially bounded generalized eigenfunction $\psi$. We refer to $(\theta^*,E^*,\psi)$ as a generalized eigentriple.
	
	Recall that $\delta=\rho-\alpha\beta>0.$ 
	Fix a  number $c$ such that
	\begin{equation}\label{xuanze}
		0<c<\frac{1}{10^4},
		\quad
		(1-4c)(\rho-\frac34\delta)-\alpha\beta
		>\frac15\delta.
	\end{equation}
	Such a choice is possible, since, as $c\to 0$, the left-hand side
	of the second inequality tends to
	\[
	\rho-\frac34\delta-\alpha\beta
	=\frac14\delta
	>\frac15\delta.
	\]
	
	Recall that $q_k$ denotes the denominator of the $k$-th
	continued-fraction approximation of $\omega$. It is enough to prove
	\eqref{aim} on the annuli
	\[
	\mathcal J_k
	:=
	[\lfloor\sqrt{c}\,q_k\rfloor;
	\lfloor\sqrt{c}\,q_{k+1}\rfloor]
	\bigcup
	[-\lfloor\sqrt{c}\,q_{k+1}\rfloor;
	-\lfloor\sqrt{c}\,q_k\rfloor]
	\]
	for all sufficiently large $k$, since these annuli cover $\Z$ except 
	finitely many integers. We also denote the enlarged annuli by 
	\[
	\widetilde{\mathcal J}_k
	:=
	[\lfloor cq_k\rfloor;
	\lfloor\frac{1}{10}q_{k+1}\rfloor]
	\bigcup
	[-\lfloor\frac{1}{10}q_{k+1}\rfloor;
	-\lfloor cq_k\rfloor].
	\]
	
In the following, we fix some sufficient large $k$.  To estimate  $|\psi(z)|$ for $z\in\mathcal J_k$, we  will  establish
	regularity for every point $z\in\widetilde{\mathcal J}_k$ to  associate
	with it a good interval.
	The regularity parameter and the   length  of the good interval will
	depend on whether $z$ lies within distance $cq_k$ of $q_k\Z$. As in \cite{JL18,JK24}, we define the following  resonant and nonresonant regions.
	
	For each $l\in\Z$, define the $l$-th
	\textbf{resonant region} by
	\[
	\mathcal R_{k,l}
	:=
	\{z\in\Z:\operatorname{dist}(z,lq_k)\leq cq_k\},
	\]
	and define the \textbf{nonresonant region} by
	\[
	\mathcal{NR}_k
	:=
	\widetilde{\mathcal J}_k
	\setminus
	\bigcup_{|l|\geq1}\mathcal R_{k,l}.
	\]
	The following proposition establishes different  types of 
	regularity for points in the nonresonant region and the resonant regions. 
	\begin{prop}\label{regu}
		For any  generalized eigentriple $(\theta^*,E^*,\psi)$,
		$\iota>0$ and all   sufficiently large $k$ (depending on $\iota, \rho,\delta,\alpha,\omega, c,\psi$),  the following statements hold.
		\begin{itemize}
			\item[(1)] Let $z\in\mathcal{NR}_k$. Set
			\begin{align}
				\widetilde{k}
				&:=
				\max\{t\in \Z^+:q_t<cq_k\},
				\label{sk}\\
				s'
				&:=
				\max\{t\in\Z^+:tq_{\widetilde{k}}<cq_k\},
				\label{sk2}\\
				\widetilde{s}
				&:=
				\max\{1,\lfloor\frac{s'}{10}\rfloor\}.
				\label{swan}
			\end{align}
			Then $z$ is $(q_{\widetilde{k}},2\widetilde{s}q_{\widetilde{k}}-1,\theta^*,E^*,\iota)$-regular.
			\item[(2)] Recall that
			$
			\beta_k:=\frac{\ln q_{k+1}}{q_k}.
			$
			Let $z\in\mathcal R_{k,l}$ with
			$
			1\leq |l|
			\leq
			\lfloor\frac{q_{k+1}}{10q_k}\rfloor.
			$
			Then $z$ is
			$
			(q_k,2q_k-1,\theta^*,E^*,
			\frac12\alpha
			(\beta_k-\frac{\ln|l|}{q_k})+\iota)
			$-regular. If
			$\lfloor\frac{q_{k+1}}{10q_k}\rfloor=0$, this statement is
			understood to be empty.
		\end{itemize}
	\end{prop}
	\begin{rem}
		Proposition \ref{regu} describes  that in the nonresonant  region (the main $(1-2c)$-portion of $\widetilde{\mathcal J}_k$),  the
		regularity parameter can be arbitrarily small with  the corresponding good
		intervals having length of order $cq_k$;  in the
		resonant regions (a small $2c$-portion of $\widetilde{\mathcal J}_k$), the
		regularity parameter depends  on $\alpha$ and the continued-fraction properties of
		$\omega$ (eventually bounded by $\frac12\alpha\beta$)  with the corresponding good
		intervals having length  $2q_k-1$. It can be viewed as  a unified formulation of the results in \cite[Subsection 5.2]{JK24}.  Here we combine the cases that admit a common treatment. Our proof largely follows the approach of \cite{JK24}, where the interpolation argument originating in \cite{Jit94} for the cosine potential was extended to the setting of monotone potentials. In the next subsection, we first outline the key ideas of the interpolation argument, with particular emphasis on the setting of  monotone potentials, and then provide a complete proof of the proposition in its present unified form.
	\end{rem}

	The following lemma  recalls the  Poisson formula for solutions on a
	finite interval. In particular, when the interval is good, the value at
	an interior point can be controlled by an exponentially weighted norm of
	the solution outside the interval.
	\begin{lem}
		Let $\theta\in \T$, $E\in \R$ and  $\varphi $  be a solution of $H(\theta)\varphi=E\varphi$.  Let $I\subset \Z$ be an integer interval and assume that $H_I(\theta)-E$ is invertible. Then for any   $z\in I$, we have  
		\begin{equation}\label{poi}
			\varphi(z) =-\varepsilon \sum_{y\in I, x\notin I}   G_I^{\theta,E}(z,y) \varphi( x)\phi_{y-x}.
		\end{equation}
		Moreover, if $I$ is  $(n,\theta ,E,\mu   )$-good  and $z$ satisfies 	$\operatorname{dist}(z, \Z\setminus  I)\geq   \frac15 n$, then for all sufficiently large $n$ (depending on $\delta$ and $\varepsilon$),  we have 
		\begin{equation}\label{yici}
			|\varphi(z)|\leq  e^{\mu  n }  \sup_{x\notin I}  (  e^{-(\rho-\frac23  \delta) |z-x|}  |\varphi (x)| ) . 
		\end{equation}
	\end{lem}
	\begin{proof}
		\eqref{poi} is the  long-range version of the  Poisson formula. 	Restricting     $(H(\theta) -E)\varphi=0$  to $I$, we have  $$(H_I(\theta)-E)\varphi =-R_I (H(\theta)-E)R_{I^c}\varphi= -\varepsilon R_I T_\phi R_{I^c}\varphi. $$ Multiplying the above equation by $  G_I^{\theta,E}$ and valuing  at $z$, we obtain \eqref{poi}. 
		
		Now assume $I$ is  $(n,\theta ,E,\mu   )$-good  and $z$ satisfies 	$\operatorname{dist}(z, \Z\setminus  I)\geq   \frac15 n$. Substituting \eqref{lr} and  \eqref{ig} to \eqref{poi}, we obtain 
		\begin{align}
			\nonumber|	\varphi(z)| &\leq \varepsilon \sum_{y\in I, x\notin I} e^{(\frac{1}{15}\delta+\mu) n }    e^{-(\rho-\frac{1}{4}\delta)|z-y|}  e^{-\rho |y-x|} |\varphi( x)|. \\
			&\nonumber\leq\varepsilon     e^{(\frac{1}{15}\delta+\mu) n } n \sum_{x\notin I}  e^{-(\rho-\frac14 \delta) |z-x|}  |\varphi (x)| \\
			&\leq  e^{(\frac{1}{15}\delta+\mu) n }  (\varepsilon  n  \sum_{x\notin I}   e^{-\frac{1}{12} \delta |z-x|}  )     \sup_{x\notin I}  ( e^{-(\rho-(\frac14+ \frac{1}{12}) \delta) |z-x|}  |\varphi (x)| ).\label{123456}
		\end{align} 
		Since $\operatorname{dist}(z, \Z\setminus  I)\geq   \frac15 n$, we have 
		\begin{equation}\label{1031}
			e^{\frac{1}{15}\delta n}\leq  e^{\frac{1}{3}\delta |z-x|} \quad \text{for $x\notin I$}
		\end{equation}  and for sufficiently large $n$,  
		\begin{equation}\label{1032}
			\varepsilon n \sum_{x\notin I}   e^{-\frac1{12} \delta |z-x|}\leq \varepsilon n \sum_{|j|\geq \frac{1}{5}n}    e^{-\frac1{12} \delta |j|}  \leq 1.  
		\end{equation}  Substituting \eqref{1031} and \eqref{1032} to \eqref{123456}, we obtain \eqref{yici}. 
	\end{proof} 
	By Propositions \ref{2116} and \ref{regu}, we can associate every point $z\in\widetilde{\mathcal J}_k$ with  a  good interval $I(z)$ satisfying  $	\operatorname{dist}(z,\Z\setminus I(z))\geq\frac14 |I(z)|$.  We now show how these good intervals can be patched together to establish the exponential decay of $\psi$. 
	
	For $z$ in the nonresonant region where the regularity parameter can be arbitrarily small,  
	a  direct application of \eqref{yici} gives the usual exponential decay
	factor.
	
	However, for $z$ in a resonant region where  the regularity parameter has
	size
	$
	\frac12\alpha(\beta_k-\frac{\ln|l|}{q_k}) 
	$ and  the good intervals have length $2q_k-1$, 
	a single application of \eqref{yici} may not give the usual exponential
	decay factor. Indeed, the exponential factor $
	exp\{\frac12\alpha(\beta_k-\frac{\ln|l|}{q_k})(2q_k-1)\}
	$, which is  eventually bounded by $e^{\alpha \beta q_k}$, 
	may defeat the exponential decay factor  gained from a jump of length less than
	$(1-2c)q_k$. However, if this jump
	has length shorter than $(1-2c)q_k$, then the landing point $x$ belongs to the
	nearby nonresonant region. We can then use the good intervals in that
	nonresonant region  where the regularity parameter is small to apply Poisson formula,  until the
	iteration exits the nonresonant region. The last jump either reaches
	another resonant region at a distance of order $q_k$, or returns to the
	original resonant region after a round trip of total length close to
	$q_k$. This additional distance compensates for the previous loss. The
	following proposition makes this argument precise.

	\begin{prop}\label{21151}
		For the eigentriple $(\theta^*,E^*,\psi)$ and sufficiently large $k$  (depending on $ \rho,\delta,\alpha,\omega, c,\psi$), the
		following statements hold.
		\begin{itemize}
			
			\item[(1)] Let $z\in\mathcal{NR}_k$. Then there exists
			$x\notin\mathcal{NR}_k$ such that
			\begin{equation}\label{cs1}
				|\psi(z)|
				\leq e^{-(\rho-\frac34\delta)|z-x|}|\psi(x)|.
			\end{equation}
			
			\item[(2)] Let $z\in\mathcal R_{k,l}$ with
			$		1\leq |l|\leq \lfloor\frac{q_{k+1}}{20q_k}\rfloor .$
			Then there exists
			$x\neq z$ such that
			\begin{equation}\label{cs21}
				|\psi(z)|
				\leq
				e^{-((1-2c)(\rho-\frac34\delta)
					-\alpha(\beta_k-\frac{\ln|l|}{q_k}))|z-x|}
				|\psi(x)|.
			\end{equation}
			
		\end{itemize}
	\end{prop}

	\begin{rem}\label{rnm}
		Recalling the choice of $c$ (cf. \eqref{xuanze}), for sufficiently large
		$k$, the exponent in \eqref{cs21} satisfies
		\[
		(1-2c)(\rho-\frac34\delta)
		-\alpha(\beta_k-\frac{\ln|l|}{q_k})
		>\frac16\delta .
		\]
		Therefore, Proposition \ref{21151} implies that for sufficiently large
		$k$ and  every
		\begin{equation}\label{jhat}
			z\in 	[\lfloor cq_k\rfloor,\lfloor\tfrac1{40}q_{k+1}\rfloor]
			\bigcup
			[-\lfloor\tfrac1{40}q_{k+1}\rfloor,-\lfloor cq_k\rfloor]
			=:\widehat{\mathcal J}_k,
		\end{equation}
		there exists $x\neq z$ such that
		\begin{equation}\label{it}
			|\psi(z)|
			\leq e^{-\frac16\delta|z-x|}|\psi(x)|.
		\end{equation}
	\end{rem}
	\begin{proof}[Proof of Proposition \ref{21151}]
		We first prove (1). Let $z\in\mathcal{NR}_k$. 
		Applying  Proposition \ref{2116} and  Proposition \ref{regu} (1) with
		$\iota=\frac1{100}\delta$,  denoting
		$	n_1=2\widetilde{s}q_{\widetilde{k}}-1 ,
		$  there exists a
		$(n_1,\theta^*,E^*,\frac1{100}\delta)$-good integer interval $I(z)$ such
		that
		\[
		\operatorname{dist}(z,\mathbb Z\setminus I(z))
		\geq \frac14 n_1 .
		\]
		Hence, by \eqref{yici},
		\begin{equation}\label{1234}
			|\psi(z)|
			\leq
			e^{\frac1{100}\delta n_1}
			\sup_{x\notin I(z)}
			(
			e^{-(\rho-\frac23\delta)|z-x|}
			|\psi(x)|
			).
		\end{equation}
		Since
		\[
		\lim_{|x|\to+\infty}
		e^{-(\rho-\frac23\delta)|z-x|}|\psi(x)|=0
		\]
		by $|\psi(x)|\leq 1+|x|$, the supremum in \eqref{1234} is attained at
		some $x_1\notin I(z)$. Therefore,
		\begin{equation}\label{1yici}
			|\psi(z)|
			\leq
			e^{-(\rho-\frac23\delta)|z-x_1|
				+\frac1{100}\delta n_1}
			|\psi(x_1)|
			\leq
			e^{-(\rho-\frac34\delta)|z-x_1|}
			|\psi(x_1)|,
		\end{equation}
		where we used
		$
		|z-x_1|\geq\frac14 n_1  
		$ in the second  inequality above.  
		
		If $x_1\notin\mathcal{NR}_k$, the proof of (1) is complete. Otherwise,
		we apply \eqref{1yici} with $x_1$ in place of $z$ and obtain
		\[
		|\psi(x_1)|
		\leq
		e^{-(\rho-\frac34\delta)|x_1-x_2|}
		|\psi(x_2)|
		\]
		for some $x_2\notin I(x_1)$. Iterating this procedure $t$ times until reaching a
		point $x_{t+1}\notin\mathcal{NR}_k$   (the iteration cannot  stay in this  finite region  forever since   $e^{-(\rho-\frac{3}{4}\delta)\cdot\frac{1}{4}n_1}<1$; otherwise, 
		$|\psi(z)|=0$, which makes \eqref{cs1} trivial),   we obtain
		$$
		|\psi(z)|\leq  e^{-(\rho-\frac{3}{4}\delta)(|z-x_1|+|x_1-x_2|+\cdots+|x_t-x_{t+1}|)}|\psi(x_{t+1})| \leq   e^{-(\rho-\frac{3}{4}\delta)|z-x_{t+1}|}|\psi(x_{t+1})|. 
		$$
		This proves (1).
		
		We next prove (2). Let
		$	z\in\mathcal R_{k,l}, 	1\leq |l|\leq \lfloor\frac{q_{k+1}}{20q_k}\rfloor .$
		We first prove that there exists
		$x\notin\mathcal{NR}_k$ such that
		\begin{equation}\label{cs2}
			|\psi(z)|
			\leq
			\begin{cases}
				e^{-((1-2c-\frac1{q_k})(\rho-\frac34\delta)
					-\alpha(\beta_k-\frac{\ln|l|}{q_k}))q_k}
				|\psi(x)|,
				& x\in\mathcal R_{k,l},\\[6pt]
				e^{-(\rho-\frac34\delta)|z-x|
					+\alpha(\beta_k-\frac{\ln|l|}{q_k})q_k}
				|\psi(x)|,
				& x\notin\mathcal R_{k,l}.
			\end{cases}
		\end{equation}
		Applying   Proposition \ref{2116}  and  Proposition \ref{regu} (2) with
		$
		\iota=\frac1{100}\delta
		$, 	denoting $n_2=2q_k-1$,
		there exists a
		$	(n_2,\theta^*,E^*,	\frac12\alpha(\beta_k-\frac{\ln|l|}{q_k})+\frac1{100}\delta)$-good
		integer interval $I(z)$ such that
		\[
		\operatorname{dist}(z,\mathbb Z\setminus I(z))
		\geq \frac14 n_2 .
		\]
		Hence, by \eqref{yici},
		\begin{equation}\label{12341}
			|\psi(z)|
			\leq
			e^{
				(\frac12\alpha(\beta_k-\frac{\ln|l|}{q_k})
				+\frac1{100}\delta)n_2
			}
			\sup_{x\notin I(z)}
			(
			e^{-(\rho-\frac23\delta)|z-x|}
			|\psi(x)|
			).
		\end{equation}
		By the aforementioned reason, the supremum in \eqref{12341} is attained at some
		$x_1\notin I(z)$. Therefore,
		\begin{equation}\label{2yici}
			|\psi(z)|
			\leq
			e^{\frac12\alpha(\beta_k-\frac{\ln|l|}{q_k})n_2}
			e^{-(\rho-\frac34\delta)|z-x_1|}
			|\psi(x_1)|
			\leq
			e^{-(\rho-\frac34\delta)|z-x_1|
				+\alpha(\beta_k-\frac{\ln|l|}{q_k})q_k}
			|\psi(x_1)|.
		\end{equation}
		Since
		$$
		\operatorname{dist}(z,\mathbb Z\setminus I(z))
		\geq\frac14(2q_k-1)
		$$
		and
		$$
		z\in\mathcal R_{k,l}
		=
		[lq_k-\lfloor cq_k\rfloor,
		lq_k+\lfloor cq_k\rfloor],
		\quad
		0<c<10^{-4},
		$$
		we have
		\begin{equation}\label{raoquan}
			\operatorname{dist}(x_1,\mathcal R_{k,l})
			\geq |x_1-z|-2cq_k
			\geq
			\frac14(2q_k-1)-2cq_k .
		\end{equation}
		Consequently,	$x_1\notin\mathcal R_{k,l}$. 
		
		We now distinguish two cases.
		
		\begin{itemize}
			\item If $x_1\notin\mathcal{NR}_k$, then \eqref{cs2} (the second item) follows directly
			from \eqref{2yici}.
			
			\item Suppose that $x_1\in\mathcal{NR}_k$. Applying the proved part (1)   to $x_1$,
			there exists
			$x'\notin\mathcal{NR}_k$ such that
			\begin{equation}\label{12342}
				|\psi(x_1)|
				\leq
				e^{-(\rho-\frac34\delta)|x_1-x'|}
				|\psi(x')|.
			\end{equation}
			
			If $x'\notin\mathcal R_{k,l}$, combining
			\eqref{2yici} and \eqref{12342} gives
			\[
			|\psi(z)|
			\leq
			e^{-(\rho-\frac34\delta)|z-x'|
				+\alpha(\beta_k-\frac{\ln|l|}{q_k})q_k}
			|\psi(x')|,
			\]
			which proves \eqref{cs2} (the second item).
			
			If $x'\in\mathcal R_{k,l}$, then by \eqref{raoquan},
			\[
			|x_1-x'|
			\geq
			\frac14(2q_k-1)-2cq_k .
			\]
			Hence, by \eqref{12342},
			\begin{equation}\label{999}
				|\psi(x_1)|
				\leq
				e^{-(\rho-\frac34\delta)
					(\frac14(2q_k-1)-2cq_k)}
				|\psi(x')|.
			\end{equation}
			On the other hand, by \eqref{2yici} and
			\[
			|z-x_1|\geq\frac14n_2=\frac14(2q_k-1),
			\]
			we have
			\begin{equation}\label{555}
				|\psi(z)|
				\leq
				e^{-(\rho-\frac34\delta)
					\frac14(2q_k-1)
					+\alpha(\beta_k-\frac{\ln|l|}{q_k})q_k}
				|\psi(x_1)|.
			\end{equation}
			Combining \eqref{555} and \eqref{999}, we obtain
			\[
			|\psi(z)|
			\leq
			e^{-((1-2c-\frac1{q_k})(\rho-\frac34\delta)
				-\alpha(\beta_k-\frac{\ln|l|}{q_k}))q_k}
			|\psi(x')|.
			\]
			This proves \eqref{cs2} (the first item).
			
		\end{itemize}
		Combining all cases, we have proved \eqref{cs2}.
		
		We now use \eqref{cs2} to derive \eqref{cs21}.
		\begin{itemize}
			\item First consider the case $x\in\mathcal R_{k,l}$ in \eqref{cs2}.
			Since
			$		z,x\in\mathcal R_{k,l},$
			we have
			$	|z-x|\leq 2cq_k.$
			By the choice of $c$ (cf. \eqref{xuanze}), for sufficiently large $k$,
			$$
			\frac1{2c}
			\left[
			(1-2c-\frac1{q_k})(\rho-\frac34\delta)
			-\alpha(\beta_k-\frac{\ln|l|}{q_k})
			\right]
			>
			(1-2c)(\rho-\frac34\delta)
			-\alpha(\beta_k-\frac{\ln|l|}{q_k})
			>0 .
			$$
			Therefore,
			\begin{align*}
				|\psi(z)|
				&\leq
				e^{-((1-2c-\frac1{q_k})(\rho-\frac34\delta)
					-\alpha(\beta_k-\frac{\ln|l|}{q_k}))q_k}
				|\psi(x)|
				\\
				&\leq
				e^{-((1-2c-\frac1{q_k})(\rho-\frac34\delta)
					-\alpha(\beta_k-\frac{\ln|l|}{q_k}))
					\frac{|z-x|}{2c}}
				|\psi(x)|
				\\
				&\leq
				e^{-((1-2c)(\rho-\frac34\delta)
					-\alpha(\beta_k-\frac{\ln|l|}{q_k}))
					|z-x|}
				|\psi(x)|.
			\end{align*}
			Moreover, the factor
			\[
			e^{-((1-2c-\frac1{q_k})(\rho-\frac34\delta)
				-\alpha(\beta_k-\frac{\ln|l|}{q_k}))q_k}
			<1
			\]
			excludes the possibility $x=z$; otherwise, 
			$|\psi(z)|=0$, which makes \eqref{cs21} trivial.
			
			\item Now consider the case $x\notin\mathcal R_{k,l}$ in \eqref{cs2}.
			Since $x\notin\mathcal{NR}_k$, either
			$	x\in\mathcal R_{k,l'}	$ for some $l'\neq l	$ or
			$
			x\notin\widetilde{\mathcal J}_k .
			$
			Because
			$
			z\in\mathcal R_{k,l},
			$ $	1\leq |l|\leq
			\lfloor\frac{q_{k+1}}{20q_k}\rfloor,
			$
			both cases imply
			\[
			|z-x|\geq (1-2c)q_k .
			\]
			Hence,
			\begin{align*}
				|\psi(z)|
				&\leq
				e^{-(\rho-\frac34\delta)|z-x|
					+\alpha(\beta_k-\frac{\ln|l|}{q_k})q_k}
				|\psi(x)|
				\\
				&\leq
				e^{-(\rho-\frac34\delta)|z-x|
					+\alpha(\beta_k-\frac{\ln|l|}{q_k})
					\frac{|z-x|}{1-2c}}
				|\psi(x)|
				\\
				&\leq
				e^{-((1-2c)(\rho-\frac34\delta)
					-\alpha(\beta_k-\frac{\ln|l|}{q_k}))
					|z-x|}
				|\psi(x)|.
			\end{align*}
			This completes the proof of (2).
		\end{itemize}
		
	\end{proof}
	Now we are ready to prove Theorem \ref{main}. 
	\begin{proof}[Proof of Theorem \ref{main}]
		By the analysis at the beginning of this subsection, Anderson localization follows if one  proves that for any eigentriple $(\theta^*,E^*,\psi)$, there exists  some $\rho'>0$,  \begin{equation}\label{yz}
			|\psi(z)|\leq     e^{-\rho'|z|}  
		\end{equation}  holds  for  $$z\in \mathcal{J}_k= [\lfloor \sqrt{c}q_k\rfloor; \lfloor \sqrt{c} q_{k+1}\rfloor ]\bigcup   [-  \lfloor \sqrt{c} q_{k+1}\rfloor   ; - \lfloor \sqrt{c}q_k\rfloor ] $$ for all  sufficiently large $k$. Let  $z\in  \mathcal{J}_k.$ Recalling the notation $  \widehat{ \mathcal J}_k$ (cf. \eqref{jhat}), we also have $z\in \widehat{ \mathcal J}_k $. By Remark \ref{rnm},  we can iterate \eqref{it}  $t$ times to obtain 
		$$	|\psi(z)|\leq e^{-\frac{1}{6}\delta |z-x_1|}|\psi(x_1)|\leq e^{-\frac{1}{6}\delta (|z-x_1|+|x_1-x_2|)}|\psi(x_2)|\leq \cdots \leq  e^{-\frac{1}{6}\delta (|z-x_1|+\cdots+|x_{t-1}-x_{t}|)}|\psi(x_t)| .$$
		until $x_{t}\notin \widehat{ \mathcal J}_k$ (the iteration cannot stay in this finite region forever, since   $e^{-\frac{1}{6}\delta }<1$).  
		Therefore,  $$
		|\psi(z)|\leq  \sup_{x\notin \widehat{ \mathcal J}_k } e^{-\frac{1}{6}\delta |z-x|} |\psi(x)| \leq  \sup_{x\notin \widehat{ \mathcal J}_k }  e^{-\frac{1}{6}\delta |z-x|}(1+|x|) .$$
		For $x\notin \widehat{ \mathcal J}_k$, we have  either $|x|\leq cq_k$ or $|x|\geq \frac 1{40}q_{k+1}$.  \begin{itemize}
			\item  For $|x|\leq cq_k$, since $|z|\geq \sqrt{c}q_k$,  we have $|x-z|\geq |z|-|x|\geq (1-\sqrt{c})|z|$. Hence,  sufficiently large for   $k$,   $$ e^{-\frac{1}{6}\delta|z-x|}(1+|x|)
			\leq e^{-\frac{1}{6}(1-\sqrt c)\delta|z|}(1+|z|)
			\leq e^{-\frac17\delta|z|}. $$ 
			\item  For $|x|\geq \frac 1{40}q_{k+1}$, since $|z|\leq \sqrt{c}q_{k+1}\leq \frac{1}{100}q_{k+1}$,  we have $|x|\geq2|z|$ and thus $|x-z|\geq |x|/2$.   Hence,  sufficiently large for   $k$,   $$ e^{-\frac{1}{6}\delta|z-x|}(1+|x|)
			\leq e^{-\frac{1}{12}\delta|x|}(1+|x|)
			\leq e^{-\frac{1}{14}\delta|x|}
			\leq e^{-\frac17\delta|z|}. $$ 
		\end{itemize}
		This proves \eqref{yz} with $\rho'=\delta/7$.  
		Finally, we obtain the uniform Anderson localization \eqref{unif} by a 
		phase-shifting argument. Indeed, for any normalized eigenfunction 
		$\psi$, let $x_\psi\in \Z$ such that  $|\psi(x_\psi)|=\|\psi\|_{\ell^\infty}=1$. Since the above argument is uniform in $\theta^*$ and $E^*$, 
		by shifting the phase by $x_\psi \omega $, we may assume that
		$x_\psi =0$. Moreover, one can see from the proof of Proposition \ref{regu}  that the largeness condition on $k$  is  independent of
		$\theta^*$, $E^*$, and $\psi$ satisfying  $|\psi(0)|=\|\psi\|_{\ell^\infty}=1$.  Consequently, there exists
		$K>0$, independent of
		$\theta^*$, $E^*$, and $\psi$ satisfying  $|\psi(0)|=\|\psi\|_{\ell^\infty}=1$, such that
		$|\psi(x)|\le e^{-\delta|x|/7}$ for all $|x|\ge K$.
		The remaining finitely many sites are covered by taking
		$C=e^{\delta K/7}$. This proves \eqref{unif}.
	\end{proof}
	
	\subsection{Proof of the regularity result}  
	The proof of Proposition \ref{regu} relies on an interpolation argument which traces back to  \cite{Jit94,Jit99} with further developments in \cite{JL18,JL24,JK24}. Before presenting the proof, we first outline the general idea of this  argument and then explain how it applies in the present setting of monotone potentials. 
	
	The argument has two steps. The first step uses a Lagrange interpolation type argument to guarantee the existence of a good point among two integer intervals (One  is associated with the origin, and the other is associated with the target point $z$, according to the definition of regularity (cf. \eqref{ac})). The second step uses the generalized eigenfunction $\psi$
	and the Poisson formula to exclude the integer interval associated with the  origin. Consequently,
	the good point must lie in  that associated with the target point
	$z$, which gives the desired regularity of $z$.
	
	In the present setting, for each $z$ under consideration, we choose two disjoint
	integer intervals $\mathcal I_1$ and $\mathcal I_2$ of  length $\frac{n+1}{2}$, where $n=2\widetilde{s}q_{\widetilde{k}}-1$ or $2q_{k}-1$  is the relevant  scale. 
	The first one is  near the origin, while the second one is its
	translate near $z$. Then 
	$	
	\mathcal I:=\mathcal I_1\cup\mathcal I_2
	$
	contains exactly $n+1$ integers. We regard
	$
	\{\theta+x\omega:\ x\in\mathcal I\}\subset\mathbb T
	$
	as a family of phase sampling points.  
	Now the problem becomes to find some sampling point $\theta+x\omega$ such that   $|P_n^{<b}(\theta+x\omega ,E)|$ has the corresponding exponential lower bound.  Recall the quasi-zeros (cf. Definition \ref{qz}). By  the  estimate \eqref{sf},  we can replace $|P_n^{<b}(\cdot,E)|$ by a product of distances from the phase  to certain quasi-zeros.  Using the continued-fraction properties of 
	$\omega$, one can establish two crucial properties for the family of sampling points:
	\begin{itemize}
		\item[(1)] These   sampling points are mutually $d$-separated on $\T$,  
		where in  the nonresonant case one has
		$d\gtrsim q_k^{-1}$, and in the resonant case
		$\mathcal R_{k,l}$  one has 
		$
		d\gtrsim \frac{|l|}{q_{k+1}}.
		$
		\item[(2)] These  sampling points can be divided into several groups with the same number of sampling points in each group. Sampling points in the same  group are located in  one  small interval on $\T$ and the corresponding intervals for different groups are mutually separated. 
	\end{itemize}
	Since there are only $n$ quasi-zeros for the 
	scale $n$, while there are $n+1$ sampling points, a pigeonhole argument
	produces a small interval  containing more  sampling
	points than  quasi-zeros in its certain enlargement. In this  small interval, the distance on $\T$ coincides with the usual distance and the number of the effective quasi-zeros is  less than that of the  sampling
	points.    Using  Lagrange interpolation formula and separation of sampling points in this small interval,  one obtains a sampling point
	$\theta+x\omega$ at which the  polynomial produced by the product of  distances to these effective quasi-zeros (hence $|P_n^{<b}(\cdot,E)|$) has the desired lower bound. Equivalently, at least one point
	$x\in\mathcal I_1\cup\mathcal I_2$ is good. 
	
	It remains to show  the good point lies in $\mathcal I_2$, or equivalently,  no point of $\mathcal I_1$ can be good. Otherwise, there would be a sequence of  increasing   good intervals around the origin.
	Applying the Poisson formula argument from  the previous subsection on these intervals would force
	$\psi(w)\to0$ at every fixed $w\in\mathbb Z$, which
	contradicts the assumption that $\psi$ is a nonvanishing generalized
	eigenfunction.
	
	Now we give the complete proof. 
	\begin{proof}[Proof of Proposition \ref{regu} (1)]
		Fix $\iota >0$. Let $z \in \mathcal{NR}_k$. By the choice of $\widetilde{k}$ and $\widetilde{s}$ (cf. \eqref{sk}, \eqref{sk2} and  \eqref{swan}),  we have $\widetilde{s}q_{\widetilde{k}}\leq \lfloor cq_k\rfloor$. Since $|z|\geq \lfloor cq_k\rfloor $, the two integer  intervals $$\mathcal I_1=[0;\widetilde{s}q_{\widetilde{k}}-1] -\lfloor \frac34 (2\widetilde{s}q_{\widetilde{k}}-1)\rfloor, \quad \mathcal I_2=[0;\widetilde{s}q_{\widetilde{k}}-1]-\lfloor \frac34 (2\widetilde{s}q_{\widetilde{k}}-1)\rfloor  +z=\mathcal{I}_1+z$$
		are disjoint. Hence $\mathcal I=\mathcal I_1\cup \mathcal I_2$ contains $2\widetilde{s}q_{\widetilde{k}}$ points. We first establish a  lower bound of  $\|(x-y)\omega\|_\T $ for $x\neq y\in \mathcal{I}$ to use interpolation argument later.  Note that  $\mathcal I_1$ and  $\mathcal I_2$ are integer intervals  of length not exceeding  $\widetilde{s}q_{\widetilde{k}}\leq \lfloor cq_k\rfloor$  and   $z\in  \mathcal{NR}_k$  satisfies $\operatorname{dist}(z,q_k\Z)\geq \lfloor cq_k\rfloor$.     Thus  for any $x,y\in \mathcal{I}, x\neq y$, we have  $x-y\notin q_k\Z$.  Write  $$x-y=tq_k+r \quad \text{for some }\quad 1\leq |r|\leq q_k-1,\ \  0\leq |t|\leq \lfloor \frac{q_{k+1}}{10q_k}\rfloor.$$
		Hence by \eqref{best}, we have 
		\begin{equation}\label{fenl}
			\|(x-y)\omega\|_\T\geq \|r\omega\|_\T- |t|\|q_k\omega\|_\T\geq \frac{1}{2q_k}-	\lfloor \frac{q_{k+1}}{10q_k}\rfloor \frac{1}{q_{k+1}} \geq \frac{1}{4q_k},\quad x,y\in \mathcal{I}, x\neq y.
		\end{equation}
		For a set $A\subset \Z$ and a phase $\theta_0\in \T$, denote $$\theta_0+ A\omega:=\{\theta_0+ a\omega \text{  modulo $ \Z$} :\ a\in A\}\subset\T .$$
		We first claim that 
		\begin{claim}\label{cl}
			For any  $\theta\in \T$, $E\in \R$ and sufficiently  large $k$ (depending on $\iota ,c, \alpha $), $\mathcal{I}$ has at least one   $(q_{\widetilde{k}},2\widetilde{s}q_{\widetilde{k}}-1,\theta ,E,\iota)$-good point.
		\end{claim}
		\begin{proof}[Proof of Claim \ref{cl}]
			
			For the given $\theta$, we will call  $\theta+\mathcal I\omega $ the  sampling points below.   	By \eqref{fenl},  these sampling points are  $\frac{1}{4q_k}$-separated, that is, 
			\begin{equation}\label{ddd}
				\|\theta_i-\theta_{i'}\|_\T\geq \frac{1}{4q_k} , \quad  \theta_i, \theta_{i'}\in \theta+\mathcal I\omega, \ \   \theta_i\neq\theta_{i'}.
			\end{equation}
			
			We distinguish two cases.
			
			\underline{Case 1: $\widetilde{s}=1$.} In this case, by \eqref{swan}, we have  $\lfloor \frac{s'}{10}\rfloor\leq \widetilde{s}=1  $, hence  $s'<20$. By the maximality  of $s'$ (cf. \eqref{sk2}), we have  $$
			20q_{\widetilde{k}}\geq cq_k. 
			$$
			Therefore, provided $\widetilde{k}$ (equivalently, $k$)  sufficiently  large (depending on $c$),  we have  $$
			b(q_{\widetilde{k}})^{1/\alpha}=q_{\widetilde{k}}^{-10}< \frac{c}{200q_{\widetilde{k}}} \leq \frac{1}{10q_k}.
			$$
			Since there are $n=2q_{\widetilde{k}}-1$ quasi-zeros (cf. Definition \ref{qz}), it follows from  Pigeonhole principle and the $\frac{1}{4q_k}$-separation of sampling  points   that there exists some sampling point $\theta+ x\omega $ ($x\in \mathcal I$) such that 
			$$
			\min_{0\leq k\leq 2q_{\widetilde{k}}-2}\|\theta+ x\omega -\widetilde{z}_k(E)\|_\T \geq \frac{1}{10q_k}>	b(q_{\widetilde{k}})^{1/\alpha} . 
			$$
			Therefore, setting $n=2q_{\widetilde{k}}-1     ,b=b(q_{\widetilde{k}})$ in \eqref{sf} and noting the convention after that, we have 
			$$	|P_{2q_{\widetilde{k}}-1}^{<b(q_{\widetilde{k}})}(\theta+x\omega ,E)|
			\geq 1 \geq e^{-\iota (2q_{\widetilde{k}}-1)}.$$ Hence, $x$ is  $(q_{\widetilde{k}},2q_{\widetilde{k}}-1,\theta ,E,\iota)$-good. 
			
			\underline{Case 2: $\widetilde{s}>1$.}  In this case, by \eqref{swan}, we have $\widetilde{s}= \lfloor \frac{s'}{10}\rfloor$ and  $s'\geq 20$. By the maximality of $\widetilde{k}$ (cf. \eqref{sk}),  we have $$q_{\widetilde{k}+1}\geq cq_k>s'q_{\widetilde{k}}
			\geq 10 \widetilde{s}q_{\widetilde{k}} .$$ Therefore, for $0\leq t\leq \widetilde{s}-1$, by \eqref{best},  we have \begin{equation}\label{jizhong}
				\|tq_{\widetilde{k}}\omega\|_\T\leq t\|q_{\widetilde{k}}\omega\|_\T\leq \frac{t}{q_{\widetilde{k}+1}} \leq \frac{\widetilde{s}}{10\widetilde{s}q_{\widetilde{k}}}\leq \frac{1}{10q_{\widetilde{k}}}. 
			\end{equation}
			By \eqref{good}, we have 
			$$	|\omega-\frac{p_{\widetilde{k}}}{q_{\widetilde{k}}}|<\frac{1}{q_{\widetilde{k}}q_{\widetilde{k}+1}}.
			$$ Hence, for  $ 0\leq r\leq q_{\widetilde{k}}-1$, we have 
			\begin{equation}\label{qudai}
				\|r\omega-\frac{rp_{\widetilde{k}}}{q_{\widetilde{k}}}\|_\T<\frac{1}{q_{\widetilde{k}+1}}\leq \frac{1}{10q_{\widetilde{k}}}.
			\end{equation}
			It follows from \eqref{jizhong} and \eqref{qudai} that for any fixed $ 0\leq r\leq q_{\widetilde{k}}-1$,
			\begin{equation}\label{159}
				\|(tq_{\widetilde{k}}+r)\omega-\frac{rp_{\widetilde{k}}}{q_{\widetilde{k}}}\|_\T< \frac{1}{5q_{\widetilde{k}}}, \quad  0\leq t\leq \widetilde{s}-1.
			\end{equation}
			Note that $\frac{rp_{\widetilde{k}}}{q_{\widetilde{k}}} (0\leq r\leq q_{\widetilde{k}}-1)$ is a permutation of  $\frac{j}{q_{\widetilde{k}}}   (0\leq j\leq q_{\widetilde{k}}-1)$  modulo $\Z$. Denote by $[a,b]_\T$ the interval $[a,b]$ modulo $\Z$, which can also  be viewed as an interval on $\T$ for convenience. \eqref{159} implies that, 
			each  interval $[\frac{j}{q_{\widetilde{k}}}-  \frac{1}{5q_{\widetilde{k}}},\frac{j}{q_{\widetilde{k}}} +\frac{1}{5q_{\widetilde{k}}} ]_\T $ $(0\leq j\leq q_{\widetilde{k}}-1)$ contains exactly $\widetilde{s}$  points of $0, \omega, \cdots, (\widetilde{s}q_{\widetilde{k}}-1)\omega $.       
			Since $\theta+\mathcal{I}_1\omega $ and $\theta+\mathcal{I}_2 \omega $ are merely  translations of   $\{0, \omega, \cdots, (\widetilde{s}q_{\widetilde{k}}-1)\omega \}$, there exist $a_1, a_2\in \T $ such that each  interval $[a_\xi +\frac{j}{q_{\widetilde{k}}}-  \frac{1}{5q_{\widetilde{k}}},a_\xi +\frac{j}{q_{\widetilde{k}}} +\frac{1}{5q_{\widetilde{k}}} ]_\T$ $(0\leq j\leq q_{\widetilde{k}}-1)$ contains exactly $\widetilde{s}$   points of   $\theta+\mathcal{I}_\xi \omega $ ($\xi=1,2$). Let $0\leq j^*\leq q_{\widetilde{k}}-1$ satisfy 
			$$\|a_1+\frac{j^*}{q_{\widetilde{k}}} -a_2\|_\T\leq \frac{1}{2q_{\widetilde{k}}}.$$ 
			For $0\leq j\leq q_{\widetilde{k}}-1$, let  $J_j\subset \T$ be the smallest interval containing $$[a_1  +\frac{j^*+j}{q_{\widetilde{k}}}-  \frac{1}{5q_{\widetilde{k}}},a_1  +\frac{j^*+j}{q_{\widetilde{k}}} +\frac{1}{5q_{\widetilde{k}}} ]_\T\bigcup [a_2 +\frac{j}{q_{\widetilde{k}}}-  \frac{1}{5q_{\widetilde{k}}},a_2 +\frac{j}{q_{\widetilde{k}}} +\frac{1}{5q_{\widetilde{k}}} ]_\T. $$ Then each $J_j$ satisfies  $|J_j|\leq \frac{1}{2q_{\widetilde{k}}}+ \frac{2}{5q_{\widetilde{k}}}=\frac{9}{10q_{\widetilde{k}}} $ and  contains exactly $2\widetilde{s}$ sampling points of $\theta+\mathcal I\omega  $. Moreover,  $J_j$ and $J_{j'}$ are $\frac{1}{10q_{\widetilde{k}}}$-separated for $j\neq j'$. 
			Let $\widetilde{ J}_j$ be the $q_{\widetilde{k}}^{-10}$-enlargement of $J_j$. Since $2 q_{\widetilde{k}}^{-10}<\frac{1}{10q_{\widetilde{k}}}$, we still have $\widetilde{ J}_j\cap \widetilde{ J}_{j'}=\emptyset$   for $j\neq j'$. Since there are $n=2\widetilde{s}q_{\widetilde{k}}-1$ quasi-zeros and $ q_{\widetilde{k}} $ disjoint intervals $\widetilde{ J}_j$,  it follows from the   Pigeonhole principle  that  there exists some $j_0$ such that $\widetilde{ J}_{j_0}$ contains $t\leq 2\widetilde{s}-1$ quasi-zeros. 
			Denote by\begin{equation}\label{sp}
				\text{ $\theta_1, \cdots,\theta_{2\widetilde{s}} \in \theta+\mathcal I\omega$   the $2\widetilde{s}$  sampling points in $J_{j_0}$.}
			\end{equation}  
	
			We  use the following interpolation lemma from \cite{JK24} to avoid abnormally  small value of 	$|	P_{2\widetilde{s}q_{\widetilde{k}}-1}^{<b(q_{\widetilde{k} })}(\theta_i,E)|$ among the collection of $\frac{1}{4q_k}$-separated sampling points $\theta_1, \cdots,\theta_{2\widetilde{s}}$. For the reader's convenience, we present the proof here. 
			\begin{lem}\label{lacha}
				Let $J\subset\mathbb T$ be an interval satisfying	$|J|<1/10$ and let   $b^{1/\alpha}<1/10$.  Assume that $J$ contains $s$ sampling  points $\theta_1,\cdots,\theta_s$ with 
				\begin{equation}\label{dd}
					\|\theta_i-\theta_{i'}\|_\T\geq d , \quad i\neq i',
				\end{equation}    and  the $b^{1/\alpha}$-enlargement of $J$, denoted by  $\widetilde{ J}$,  contains  $t\leq s-1$ quasi-zeros  $\widetilde{z}_{k_1}(E),\cdots,\widetilde{z}_{k_t}(E)$, that is,   
				\begin{align} 	 \inf_{a\in J}\|\widetilde{z}_k-a \|_\T=:\operatorname{dist}_\T(\widetilde{z}_k,J)&\leq b^{1/\alpha}, \quad k= k_1,\cdots,k_t,\\ 
					\label{lie}
					\operatorname{dist}_\T(\widetilde{z}_k,J)&>b^{1/\alpha}, \quad k\neq k_1,\cdots,k_t.
				\end{align}	Then there exists $i\in\{1,\cdots,s\}$ such that
				\begin{equation}\label{2123}
					|	P_n^{<b}(\theta_i,E)|
					\geq
					\left[
					(\frac d2)^t t!
					\right]^\alpha.
			\end{equation}\end{lem}

			\begin{proof}[Proof of Lemma \ref{lacha}] 
				Recall \eqref{sf} that 
				\begin{equation}\label{sf2}
					|P_n^{<b}(\theta,E)|
					\geq
					\prod_{\substack{0\le k\leq n-1\\
							\|\theta-\widetilde z_k(E)\|_{\mathbb T}<b^{1/\alpha}}}
					\|\theta-\widetilde z_k(E)\|_{\mathbb T}^{\alpha}.
				\end{equation}
				By \eqref{lie}, for $\theta\in J$,    every quasi-zero contributing to the
				right-hand side of \eqref{sf2} lies in  $\widetilde{ J}$.   
				Since $|\widetilde{ J}|<|J|+2b^{1/\alpha }<1/2$,   we may  view it as an interval on  $\R$ and use ordinary distance
				for points in $\widetilde{ J}$.
				After identifying $J$ with an interval in $\R$ and $\widetilde{z}_{k_1}(E),\cdots,\widetilde{z}_{k_t}(E)$ with points in $\R$, put
				$$
				p(\theta)
				=
				\prod_{j=1}^t
				(\theta -\widetilde z_{k_j}(E)), \quad \theta\in J.
				$$
				Some of the quasi-zeros contained in $\widetilde{ J}$  may fail to satisfy
				\[
				|\theta-\widetilde z_{k_j}(E)|<b^{1/\alpha}.
				\]
				We may nevertheless insert their corresponding factors, because
				\[
				|\theta-\widetilde z_{k_j}(E)|^\alpha
				\le
				|\widetilde{ J}|^\alpha<1.
				\]
				Adding such factors can only decrease the right-hand side of \eqref{sf2}.   Therefore
				\begin{equation}\label{2126}
					|P_n^{<b}(\theta,E)|
					\geq
					\prod_{j=1}^t
					|\theta-\widetilde z_{k_j}(E)|^\alpha
					=
					|p(\theta)|^\alpha, \quad \theta\in J.
				\end{equation}
				Choose any  $t+1\leq s $ of the sampling points and reorder  them  increasingly $\theta_0<\theta_1<\cdots<\theta_t$.  The Lagrange interpolation formula and comparison of the leading	coefficients give
				\[
				1
				=
				\sum_{i=0}^t
				\frac{p(\theta_i)}
				{\prod_{ j\neq i}(\theta_i-\theta_j)}.
				\]
				\eqref{dd}  implies
				$$ |\theta_i-\theta_j|\geq |i-j|d, \quad 
				|
				\prod_{ j\neq i}(\theta_i-\theta_j)
				|
				\geq
				d^t i!(t-i)!.
				$$
				Therefore
				$$
				1
				\le
				\max_i|p(\theta_i)|
				\,d^{-t}
				\sum_{i=0}^t
				\frac{1}{i!(t-i)!}
				=
				\max_i|p(\theta_i)|
				\,d^{-t}
				\frac{2^t}{t!}.
				$$
				Hence
				\begin{equation}\label{2127}
					\max_i|p(\theta_i)|
					\ge
					(\frac d2)^t t!.
				\end{equation}
				Combining \eqref{2126} and \eqref{2127} gives \eqref{2123}. This finishes the proof of Lemma \ref{lacha}.
			\end{proof}
			Setting  $n=2\widetilde{s}q_{\widetilde{k}}-1$,  $s=2\widetilde{s}$ sampling points given by \eqref{sp},  $J=J_{j_0}$, $t\leq 2\widetilde{s}-1$ quasi-zeros in  $\widetilde{ J}=\widetilde{ J}_{j_0}$,  $b= b(q_{\widetilde{k} })$, $d=\frac{1}{4q_k}$ (cf. \eqref{ddd})  in Lemma \ref{lacha},   we obtain that there exists some $\theta_i=\theta+x_i\omega $ ($x_i\in \mathcal{I}$)  satisfying   
			\begin{equation}\label{1916}
				|	P_{2\widetilde{s}q_{\widetilde{k}}-1}^{<b(q_{\widetilde{k} })}(\theta+x_i\omega,E)|
				\geq 
				\left[
				(\frac1 {8q_k} )^t t!
				\right]^\alpha.
			\end{equation}
			By the maximality  of $s'$ (cf. \eqref{sk2}), we have  $$
			20\widetilde{s} q_{\widetilde{k}}\geq (s'+1)q_{\widetilde{k}}\geq  cq_k.   
			$$	 This inequality together with   $t\leq 2\widetilde{s}-1\leq q_k$ and $t!\geq (t/e)^t$ implies  \begin{equation}\label{2022}
				(\frac1 {8q_k} )^t t!\geq (\frac1 {8q_k} )^{2\widetilde{s}} (2\widetilde{s})! \geq (\frac{c\widetilde{s}}{800\widetilde{s}q_{\widetilde{k}}})^{2\widetilde{s}}\geq  e ^{-\frac{2\ln q_{\widetilde{k}}}{q_{\widetilde{k}}} ( 2\widetilde{s}q_{\widetilde{k}}-1 )    }\geq e^{-\frac{\iota}{\alpha}  ( 2\widetilde{s}q_{\widetilde{k}}-1 ) },
			\end{equation}  provided $\widetilde{k}$ (equivalently, $k$) sufficiently large.  
			Combining \eqref{1916} and \eqref{2022} yields $x_i$ is  $(q_{\widetilde{k}},2\widetilde{s}q_{\widetilde{k}}-1,\theta ,E,\iota)$-good.  This proves Claim \ref{cl}. 
		\end{proof}
		The following lemma excludes the possibility of the good  point from Claim \ref{cl} coming from $\mathcal I_1$ for the eigentriple $(\theta^*,E^*,\psi)$ when $\iota$ is small.   
		\begin{lem}\label{nogood}
			Let $\iota <\frac{1}{100}\delta$. Then for sufficiently large $k$ (depending on $ \rho ,\delta, \alpha, c, \psi$), the points in $\mathcal I_1$ cannot be $(q_{\widetilde{k}},2\widetilde{s}q_{\widetilde{k}}-1,\theta^* ,E^*,\iota)$-good.  
		\end{lem}
		\begin{proof}[Proof of Lemma \ref{nogood}] 
			Assume the conclusion is not true. Then  there exist a sequence $\widetilde{k}_i\to \infty$ such that the origin $0$ is $(q_{\widetilde{k}_i},2\widetilde{s}_iq_{\widetilde{k}_i}-1,\theta^* ,E^*,\iota)$-regular. Hence by Proposition \ref{2116}, provided $i$  sufficiently large,  there exists a $(2\widetilde{s}_iq_{\widetilde{k}_i}-1,\theta^* ,E^*,\iota )$-good integer interval  $I_i(0)$ with $$\operatorname{dist}(0,\Z\setminus I_i(0))\geq \frac{1}{4}n_i, \quad n_i=2\widetilde{s}_iq_{\widetilde{k}_i}-1 .$$ 
	Let $	|w|\leq \frac{1}{20}n_i$. 	We have $$\operatorname{dist}(w,\Z\setminus I_i(0))\geq \frac{1}{5}n_i .$$
			Since $(\theta^*,E^*, \psi )$ is a generalized eigentriple,	by \eqref{yici},  we have 
			\begin{align*}
				|\psi (w)|&\leq  e^{\iota    n_i }  \sup_{x\notin I_i(0)}  (  e^{-(\rho-\frac23  \delta) |w-x|}  |\psi (x)| ) \\
				&\leq 	 e^{\frac{1}{100}\delta   n_i }  \sup_{|x|\geq \frac{1}{4}n_i}  (  e^{-\frac{1}{3}  \delta (|x|-\frac{1}{20}n_i)}  (1+|x|) )\\
				&\leq 	   \sup_{|x|\geq \frac{1}{4}n_i}  (e^{\frac{1}{15}\delta   |x|} \cdot   e^{-\frac{1}{3}\cdot   \frac45 \delta  |x|}  (1+|x|) )\\
				&	=	   \sup_{|x|\geq \frac{1}{4}n_i}  (   e^{-\frac{1}{5}\delta  |x|}  (1+|x|) ). 
			\end{align*}	
			Letting $i\to \infty$  yields a contradiction to   the nonvanishing assumption of $\psi$. 
		\end{proof}
		Combining Claim \ref{cl} and Lemma \ref{nogood} yields that for $0<\iota<\frac{1}{100}\delta$, $\mathcal I_2$ contains at least one $(q_{\widetilde{k}},2\widetilde{s}q_{\widetilde{k}}-1,\theta^* ,E^*,\iota)$-good point, which implies $z$ is $(q_{\widetilde{k}},2\widetilde{s}q_{\widetilde{k}}-1,\theta^* ,E^*,\iota)$-regular for any fixed $0<\iota<\frac{1}{100}\delta$ provided $k$ sufficiently large.  Note that if $\iota'>\iota$, then $z$ is also  $(q_{\widetilde{k}},2\widetilde{s}q_{\widetilde{k}}-1,\theta^* ,E^*,\iota')$-regular. This finishes the proof of Proposition \ref{regu} (1).  
	\end{proof}
	\begin{proof}
		[Proof of Proposition \ref{regu} (2)] 
		Fix $\iota >0$. Let $z \in \mathcal{R}_{k,l}$, $1\leq |l|\leq \lfloor\frac{q_{k+1}}{10q_k}\rfloor$. 
		We consider the two disjoint integer  intervals $$\mathcal I_1=[0;q_{k}-1] -\lfloor \frac34 (2q_k-1)\rfloor+z-lq_k, \quad \mathcal I_2=[0;q_{k}-1] -\lfloor \frac34 (2q_k-1)\rfloor+z=\mathcal{I}_1+lq_k. $$
		and let  $\mathcal I=\mathcal I_1\cup \mathcal I_2,$ which  contains $2q_k$ points. As in the proof for nonresonant region,  we first establish a  lower bound of  $\|(x-y)\omega\|_\T $ for $x\neq y\in \mathcal{I}$ to use interpolation argument later.   For  any $x,y\in \mathcal{I}, x\neq y$,  we can write  $$x-y=\widetilde{l}q_k+r \quad \text{for some }\quad 0\leq |r|\leq q_k-1,\ \  \widetilde{l}\in \{\pm l,0\}.$$ \begin{itemize}
			\item   If $r=0$, it must be $\widetilde{l}\neq 0$ since $x\neq y$. Since $1\leq |l|\leq \lfloor\frac{q_{k+1}}{10q_k}\rfloor$ and $\frac 1{2q_{k+1}}<\|q_k\omega  \|< \frac 1{q_{k+1}}$   (cf. \eqref{best}), we have $$\|(x-y)\omega \|_\T = \|lq_k\omega \|_\T=  |l| \|q_k\omega \| > \frac {|l|}{2q_{k+1}}.$$ 
			\item If $r\neq 0$, we have a better lower bound $\frac{1}{4q_k}$ from \eqref{fenl}
		\end{itemize}
		Therefore, 
		\begin{equation}\label{fenl2}
			\|(x-y)\omega\|_\T\geq   \frac {|l|}{2q_{k+1}}, \quad x,y\in \mathcal{I}, x\neq y.
		\end{equation}
		We first claim that 
		\begin{claim}\label{cl2}
			For any $\theta\in \T$, $E\in \R$ and  sufficiently large $k$ (depending on $\iota ,c, \alpha $), $\mathcal{I}$ has at least one   $(q_{k},2q_k-1,\theta ,E,\frac12 \alpha (\beta_k -\frac{\ln |l|}{q_k})+\iota)$-good point.
		\end{claim}
		\begin{proof}[Proof of Claim \ref{cl2}] 
			By \eqref{fenl2}, the $2q_k$ sampling points are $\frac{|l|}{2q_{k+1}}$-separated. 	The rest of the   proof  is similar to that  of Claim \ref{cl}  Case 2.
			By $|l|\leq \lfloor\frac{q_{k+1}}{10q_k}\rfloor$ and \eqref{best}, we have 
			\begin{equation}\label{jdd}
				\|lq_k\omega \|_\T\leq  |l| \|q_k\omega \|_\T \leq \frac{1}{10q_k}.
			\end{equation}
			Since $\lfloor\frac{q_{k+1}}{10q_k}\rfloor\geq 1$, we have 
			$$q_{k+1}\geq 10q_k.$$ 
			By \eqref{good}, we have 
			$$	|\omega-\frac{p_{k}}{q_{k}}|<\frac{1}{q_{k}q_{k+1}}.
			$$ Hence, for  $ 0\leq r\leq q_{k}-1$, we have 
			\begin{equation}\label{qudai2}
				\|r\omega-\frac{rp_{k}}{q_{k}}\|_\T<\frac{1}{q_{k+1}}\leq \frac{1}{10q_{k}}.
			\end{equation}
			It follows from \eqref{jdd} and \eqref{qudai2} that for any fixed $ 0\leq r\leq q_{k}-1$,
			\begin{equation}\label{1592}
				\|(tq_{k}+r)\omega-\frac{rp_{k}}{q_{k}}\|_\T< \frac{1}{5q_{k}}, \quad t\in \{0,l\}. 
			\end{equation}
			Since $\theta+\mathcal I\omega$ is merely a translation of $\{0, \omega, \cdots, (q_{k}-1)\omega \}\cup  (\{0, \omega, \cdots, (q_{k}-1)\omega \}+lq_k\omega )$, it follows from \eqref{1592} that there exist some $a\in \T$ such that for every $0\leq j\leq q_k-1$, the interval $$J_j=[a+\frac{j}{q_k}-\frac{1}{5q_k},a+ \frac{j}{q_k}+\frac{1}{5q_k}]_\T $$ contains exactly $2$ sampling points of $\theta+\mathcal I\omega$. 
			Let $\widetilde{ J}_j$ be the $q_{k}^{-10}$-enlargement of $J_j$. Since $2 q_{k}^{-10}<\frac{1}{10q_{k}}$, we still have $\widetilde{ J}_j\cap \widetilde{ J}_{j'}=\emptyset$   for $j\neq j'$. Since there are $n=2q_{k}-1$ quasi-zeros and $ q_{k} $ disjoint intervals $\widetilde{ J}_j$,  it follows from  Pigeonhole principle  that  there exists some $j_0$ such that $\widetilde{ J}_{j_0}$ contains at most one  quasi-zero. By $\frac{|l|}{2q_{k+1}}$-separation of the sampling points, at least one sampling point in $J_{j_0}$ is $\frac{|l|}{4q_{k+1}}$-separated from the possible  quasi-zero.  By \eqref{sf} and noting that  for sampling points in $ J_{j_0}$,  every quasi-zero contributing to the product in the 
			right-hand side lies in  $\widetilde{J}_{j_0}$, we obtain at least one sampling point $\theta+x\omega $ ($x\in \mathcal I$) satisfies 
			$$	|P_{2q_{k}-1}^{<b(q_{k})}(\theta+x\omega ,E)|
			\geq \left(\frac{|l|}{4q_{k+1}} \right)^\alpha=\frac1{4^\alpha}\cdot  e^{- \alpha (\beta_k -\frac{\ln |l|}{q_k})) q_{k}}, $$ which is bigger than $$  e^{-(\frac12 \alpha (\beta_k -\frac{\ln |l|}{q_k})+\iota) (2q_{k}-1)},$$ for sufficiently large $k$.   Hence, $x$ is  $(q_{k},2q_k-1,\theta ,E,\frac12 \alpha (\beta_k -\frac{\ln |l|}{q_k})+\iota)$-good. 
		\end{proof}
		The following lemma excludes the possibility of the good  point from Claim \ref{cl2} coming from $\mathcal I_1$ for the eigentriple $(\theta^*,E^*,\psi)$ when $\iota$ is small.
		\begin{lem}\label{nogood2} 
			Let  $\iota <\frac{1}{100}\delta$. Then for sufficiently large $k$ (depending on $ \rho ,\delta, \alpha, c, \psi$), the points in $\mathcal I_1$ cannot be $(q_{k},2q_k-1,\theta^* ,E^*,\frac12 \alpha \beta_k +\iota)$-good,  and hence can not be $(q_{k},2q_k-1,\theta^* ,E^*,\frac12 \alpha (\beta_k -\frac{\ln |l|}{q_k})+\iota)$-good.   
		\end{lem}
		\begin{proof}[Proof of Lemma \ref{nogood2}]
			Assume the conclusion is not true.  Then  there exist a sequence $k_i\to \infty$ such that some   point $z_i^*\in \mathcal R_{k_i,0}$  is $(q_{k_i},2q_{k_i}-1,\theta^* ,E^*,\frac12 \alpha \beta_{k_i} +\iota)$-regular. Hence by Proposition \ref{2116}, provided $i$  sufficiently large,  there exists a $(2q_{k_i}-1,\theta^* ,E^*,\frac12 \alpha \beta_{k_i} +\iota  )$-good integer interval  $I_i(z_i^*)$ with $$\operatorname{dist}(z^*,\Z\setminus I_i(z_i^*))\geq \frac{1}{4}(2q_{k_i}-1).$$ 
			Let  $w\in \mathcal{R}_{k_i,0}$. We have $$\operatorname{dist}(w,\Z\setminus I_i(z_i^*))\geq \frac{1}{4}(2q_{k_i}-1)-2cq_{k_i}.$$ By the same argument as in the proof of \eqref{cs2} (note that the regularity in  non-resonant region has been established),  we can  prove that there exists some $x\notin  \mathcal{NR}_{k_i}  $ such that 
		\begin{equation}\label{1059}
				|\psi(w)|\leq  \begin{cases}
				e^{-((1-4c-\frac{1}{q_{k_i}})(\rho-\frac{3}{4}\delta)-\alpha \beta_{k_i})  q_{k_i}}|\psi(x)|, \quad  &\text{ if } x\in  \mathcal R_{k_i,0} ,\\
				e^{-(\rho-\frac{3}{4}\delta)|w-x|+  \alpha \beta_{k_i}q_{k_i}}|\psi(x)| ,     \quad & \text{ if } x \notin  \mathcal  R_{k_i,0} .
			\end{cases} 
		\end{equation}
		 Note that \begin{equation}\label{1100}
				|w-x|\geq|x|-|w|\geq  (1-2c)|x|\geq (1-4c)q_{k_i}\quad  \text{ for $w\in \mathcal{R}_{k_i,0}$, $x\notin \mathcal  R_{k_i,0}\cup   \mathcal{NR}_{k_i}$} .
			\end{equation}  	Recall the   choice of $c$ (cf. \eqref{xuanze}).  \eqref{1059}, \eqref{1100}, together with the bound $|\psi(x)|\leq 1+|x|$ imply that    for sufficiently  large $i$,  $$	|\psi(w)|\leq e^{-\frac17\delta q_{k_i}}, \quad w\in \mathcal{R}_{k_i,0}.$$    Letting $i\to \infty$  yields a contradiction to   the nonvanishing assumption of $\psi$. 
		\end{proof}
		Combining Claim \ref{cl2} and Lemma \ref{nogood2} finishes the proof of 
		Proposition \ref{regu} (2).
	\end{proof}

	\section*{Acknowledgments}
This research was supported by the Postdoctoral Fellowship Program and the China Postdoctoral Science Foundation BX20250060.	This work was completed during a research visit to the Tianyuan Mathematical Center in Southwest China. The author gratefully acknowledges the hospitality of the Center. 
		\section*{AI-use declaration}
This research and manuscript were completed without the assistance of generative AI tools. 
		\section*{Data Availability}
	The manuscript has no associated data.
	\section*{Declarations}
	{\bf Conflicts of interest}  \ The author declares no conflict of interest.
	\normalem

\end{document}